\documentclass[12pt]{article}
\usepackage{graphicx,amsfonts,amsmath,amssymb,amsthm,mathrsfs,mathabx,url,hyperref}
\usepackage[usenames,dvipsnames]{xcolor}

\title{Another Proof of Born's Rule on Arbitrary Cauchy Surfaces}
\author{
Sascha Lill\footnote{Mathematisches Institut,
     Eberhard-Karls-Universit\"at, Auf der Morgenstelle 10, 72076
     T\"ubingen, Germany}\ \footnote{BCAM – Basque Center for Applied Mathematics, Mazarredo 14, E48009
Bilbao, Basque Country - Spain}\ \footnote{E-Mail: slill@bcamath.org, sascha.lill@uni-tuebingen.de}~\ and
Roderich Tumulka$^*$\footnote{E-mail: roderich.tumulka@uni-tuebingen.de}
} 
\date{October 14, 2021}

\usepackage[utf8]{inputenc}
\usepackage{tikz}
\usetikzlibrary{decorations.pathreplacing}

\usepackage[section]{placeins}
\usepackage{capt-of}%for captions in minipages
\usepackage{paralist}
\usepackage{booktabs}%setting distances in tables manually
\usepackage{hhline}
\usepackage{color}                  %Colors

\newcommand{\be}{\begin{equation}}
\newcommand{\ee}{\end{equation}}

\newcommand{\bx}{\boldsymbol{x}}
\newcommand{\by}{\boldsymbol{y}}

\newcommand{\cP}{\mathcal{P}}

\newcommand{\sB}{\mathscr{B}}
\newcommand{\sE}{\mathscr{E}}
\newcommand{\sH}{\mathscr{H}}

\newcommand{\sK}{\mathscr{K}}

\newcommand{\sP}{\mathscr{P}}

\newcommand{\NNN}{\mathbb{N}}
\newcommand{\MMM}{\mathbb{M}}
\newcommand{\RRR}{\mathbb{R}}
\newcommand{\ZZZ}{\mathbb{Z}}
\newcommand{\Gr}{{\rm Gr}}
\newcommand{\Sr}{{\rm Sr}}
\newcommand{\diag}{{\rm diag}}
\newcommand{\past}{{\rm past}}
\newcommand{\future}{{\rm future}}
\newcommand{\Prob}{\mathbb{P}}
\newcommand{\Hilbert}{\mathscr{H}}
\DeclareMathOperator{\range}{Ran}

\theoremstyle{plain}

\newtheorem{cor}{Corollary}
\newtheorem{lemma}{Lemma}
\newtheorem{prop}{Proposition}
\newtheorem{theorem}{Theorem}
\theoremstyle{definition} \newtheorem{definition}{Definition}

\newcounter{remarks}
\begin{document}
\maketitle
\begin{abstract}
In 2017, Lienert and Tumulka proved Born's rule on arbitrary Cauchy surfaces in Minkowski space-time assuming Born's rule and a corresponding collapse rule on horizontal surfaces relative to a fixed Lorentz frame, as well as a given unitary time evolution between any two Cauchy surfaces, satisfying that there is no interaction faster than light and no propagation faster than light. Here, we prove Born's rule on arbitrary Cauchy surfaces from a different, but equally reasonable, set of assumptions. The conclusion is that if detectors are placed along any Cauchy surface $\Sigma$, then the observed particle configuration on $\Sigma$ is a random variable with distribution density $|\Psi_\Sigma|^2$, suitably understood. The main different assumption is that the Born and collapse rules hold on any spacelike hyperplane, i.e., at any time coordinate in any Lorentz frame. Heuristically, this follows if the dynamics of the detectors is Lorentz invariant.

\medskip

\end{abstract}

\newpage
\tableofcontents

\section{Introduction}	
\label{sec:intro}

In its usual form, Born's rule asserts that if we measure the positions of all particles of a quantum system at time $t$, the observed configuration has probability distribution with density $|\Psi_t|^2$. One would expect that Born's rule also holds on arbitrary Cauchy surfaces\footnote{We use the definition that a {\bf Cauchy surface} \cite{Wald} is a subset of space-time intersected by every inextendible causal (i.e., timelike-or-lightlike) curve in exactly one point. Thus, a Cauchy surface can have lightlike tangent vectors but cannot contain a lightlike line segment.} $\Sigma$ in Minkowski space-time $\MMM$ in the following sense: If we place detectors along $\Sigma$, then the observed particle 
configuration has probability distribution with density $|\Psi_\Sigma|^2$, suitably understood. We call the latter statement the {\bf curved Born rule}; it contains the former statement as a special case in which $\Sigma$ is a horizontal 3-plane in the chosen Lorentz frame. We prove here the curved Born rule as a theorem; more precisely, we prove that the Born 
rule holds on arbitrary Cauchy surfaces assuming (i) that the Born rule holds on hyperplanes, i.e., on flat surfaces ({\bf flat Born rule}), (ii) that the collapse rule holds on hyperplanes ({\bf flat collapse rule}), (iii) that the unitary time evolution contains no interaction terms between spacelike separated regions ({\bf interaction locality}), and (iv) that 
wave functions do not spread faster than light ({\bf propagation locality}).
A similar theorem was proved by Lienert and Tumulka in \cite{LT:2020}. As we will discuss in more detail in Section 1.2, the central difference is that the detection process was modeled in a different way; our model of the detection process is in a way more natural and leads to a simpler proof of the theorem.

\bigskip

This paper is structured as follows. In the remainder of Section~\ref{sec:intro}, we describe our results. In Section~\ref{sec:definitions}, we 
provide technical details of the concepts used. In Section~\ref{sec:detection}, we derive the Born rule on triangular surfaces. In Section~\ref{sec:approxtriangular}, we prove our statements about approximating Cauchy surfaces with triangular surfaces. In Section~\ref{sec:proofthm1}, we provide the proof of our main theorem.

\subsection{Hypersurface Evolution}

In order to formulate the curved Born rule, we need to have a mathematical object $\Psi_\Sigma$ available that represents the quantum state on $\Sigma$. To this end, we regard as given a {\bf hypersurface evolution} (precise definition given in Section~\ref{sec:definitions} or \cite{LT:2020}) that provides a Hilbert space $\Hilbert_\Sigma$ for every Cauchy surface $\Sigma$ and a unitary isomorphism $U_\Sigma^{\Sigma'}:\Hilbert_\Sigma \to \Hilbert_{\Sigma'}$ representing the evolution between any two Cauchy surfaces, $\Psi_{\Sigma'}=U_\Sigma^{\Sigma'} \, \Psi_\Sigma$. The situation is similar in spirit to the Tomonaga--Schwinger approach 
\cite{tomonaga:1946,schwinger:1948,schweber:1961}, although Tomonaga and Schwinger used the interaction picture for identifying all $\Hilbert_\Sigma$ with each other.

We take the detected particle configuration on $\Sigma$ to be an element of the unordered configuration space of a variable number of particles,
\be
\Gamma(\Sigma):= \{q \subset \Sigma: \# q<\infty\}\,,
\ee
the set of all finite subsets of $\Sigma$. (If more than one, say $m\in\NNN$, species of particles are present, one may either, by straightforward generalization of our results, consider $\Gamma(\Sigma)^m$ as the configuration space or apply the mapping $\Gamma(\Sigma)^m \to \Gamma(\Sigma): (q_1,\ldots,q_m) \mapsto q_1 \cup \ldots \cup q_m$ that erases the species labels and still consider probability distributions on $\Gamma(\Sigma)$, as we will do here.)

It will be convenient to write the $|\Psi_\Sigma|^2$ distribution (the {\bf curved Born distribution}) in the form of the measure $\langle \Psi_\Sigma|P_\Sigma(\cdot)|\Psi_\Sigma\rangle=\|P_\Sigma(\cdot)\Psi_\Sigma\|^2$,  where $P_\Sigma$ is the appropriate projection-valued measure (PVM) on\footnote{We use the Borel $\sigma$-algebra on $\MMM$, $\Sigma$, $\Gamma(\Sigma)$ \cite{LT:2020} etc.; when speaking of subsets, we always mean Borel measurable subsets.} $\Gamma(\Sigma)$ acting on $\Hilbert_\Sigma$. That is, if $\Psi_\Sigma$ can be regarded as a function on $\Gamma(\Sigma)$, then, for any $S\subseteq \Gamma(\Sigma)$, $P_\Sigma(S)$ is the multiplication by the characteristic function of $S$ and
\be
\bigl\| P_\Sigma(S) \, \Psi_\Sigma \bigr\|^2 
=  \int\limits_{S} dq\, |\Psi_\Sigma(q)|^2
\ee
with $dq$ the appropriate volume measure on $\Gamma(\Sigma)$. But we do not have to regard $\Psi_\Sigma$ as a function, we can treat it abstractly 
as a vector in the given Hilbert space $\Hilbert_\Sigma$. The PVM $P_\Sigma$ is automatically given if the $\Hilbert_\Sigma$ are Fock spaces or tensor products thereof.

Another way of putting the curved Born rule (although perhaps not fully equivalent with regards to a curved collapse rule, see Remark~\ref{rem:collapse} in Section~\ref{sec:intromainresult}) is to say that {\it $P_\Sigma$ is the configuration observable on $\Sigma$}. So, our theorem could 
be summarized as showing that {\it if $P_E$ is the configuration observable on every hyperplane $E$, then $P_\Sigma$ is the configuration observable on every Cauchy surface $\Sigma$, provided interaction locality (IL) and propagation locality (PL) hold.}

A hypersurface evolution is specified by specifying the $\Hilbert_\Sigma$'s, the $U_\Sigma^{\Sigma'}$'s, and the $P_\Sigma$'s; we denote it by $(\Hilbert_\circ,U_\circ^\circ,P_\circ)$ with $\circ$ a placeholder for Cauchy surfaces. Some examples are described in \cite{LT:2020} and in Remark~\ref{rem:exampleE} in Section~\ref{sec:hypersurfaceevolution} below; they arise especially from multi-time wave functions \cite{dirac:1932,dfp:1932,schweber:1961,LT:2021}; see \cite{lienertpetrattumulka} for an introduction and overview. While certain ways of implementing an ultraviolet cutoff \cite{DN:2019,LNT:2020} lead to multi-time wave functions that cannot be evaluated on arbitrary Cauchy surfaces, models without cutoff define a hypersurface evolution, either on the non-rigorous \cite{pt:2013c,pt:2013d} or 
on the rigorous level \cite{lienert:2015a,LN:2015,DN:2016,LN:2018,KLT:2020}. As a consequence, our result proves in particular a Born rule for multi-time wave functions, thereby generalizing a result of Bloch \cite{bloch:1934} (see also Remark 4 in \cite{LT:2020}).

We do not, as one would in quantum electrodynamics or quantum chromodynamics, exclude states of negative energy; it remains for future work to extend our result in this direction.

\subsection{Previous Result}

A theorem similar to ours has been proved by Lienert and Tumulka \cite{LT:2020}; our result differs in what exactly is assumed, and how the detection process is modeled. The fact that the curved Born rule can be obtained through different models of the detection process and from different sets of assumptions suggests that it is a robust consequence of the flat 
Born rule.

In fact, our result was already conjectured by Lienert and Tumulka, who also suggested the essentials of the model of the detection process we use here, although their theorem concerned a different model. The biggest difference between their theorem and ours is that we assume the Born rule and collapse rule to hold on \emph{tilted} hyperplanes, whereas Lienert and Tumulka assumed them only on \emph{horizontal} hyperplanes in a fixed Lorentz frame.

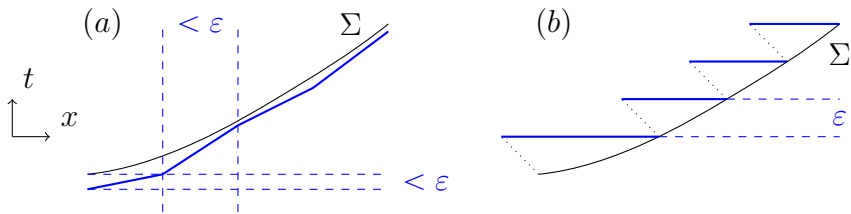
\begin{figure}[hbt]
    \centering
    \scalebox{1}{\begin{tikzpicture}

\draw[->] (-1,0.5) -- ++(0.5,0) node [anchor = south west] {$x$};
\draw[->] (-1,0.5) -- ++(0,0.5) node [anchor = south west] {$t$};

\node at (0.2,2) {$ (a) $};
\draw(0,0) .. controls (1,0.1) and (2,0.7) .. (2.5,1) .. controls (3,1.3) and (3.5,1.6).. (4,2);
\node at (3.5,2) {$ \Sigma $};

\draw[thick, blue] (0,-0.2) -- (1,0) -- (2,0.65) -- (3,1.15) -- (4,1.9);

\draw[dashed, blue] (1,-0.5) -- (1,2);
\draw[dashed, blue] (2,-0.5) -- (2,2);
\node[blue] at (1.5,2) {$ < \varepsilon $};

\draw[dashed, blue] (0,-0.2) -- (4,-0.2);
\draw[dashed, blue] (0,0) -- (4,0);
\node[blue] at (4.5,-0.1) {$ < \varepsilon $};

\node at (6.2,2) {$ (b) $};
\draw(6,0) .. controls (7,0.1) and (8,0.7) .. (8.5,1) .. controls (9,1.3) and (9.5,1.6).. (10,2);
\node at (10,1.6) {$ \Sigma $};

\draw[thick, blue] (5.5,0.5) -- (7.6,0.5);
\draw[thick, blue] (7.1,1) -- (8.5,1);
\draw[thick, blue] (8,1.5) -- (9.3,1.5);
\draw[thick, blue] (8.8,2) -- (10,2);
\draw[dotted] (6,0) -- ++(-0.5,0.5);
\draw[dotted] (7.6,0.5) -- ++(-0.5,0.5);
\draw[dotted] (8.5,1) -- ++(-0.5,0.5);
\draw[dotted] (9.3,1.5) -- ++(-0.5,0.5);

\draw[dashed, blue] (7.6,0.5) -- (10,0.5);
\draw[dashed, blue] (8.5,1) -- (10,1);
\node[blue] at (10,0.75) {$ \varepsilon $};

\end{tikzpicture}}
    \caption{(a) Our detection process is based on approximating a curved 
surface $\Sigma$ by a piecewise flat surface. (b) The detection process used by Lienert and Tumulka is based on approximating a curved surface $\Sigma$ by disconnected pieces of horizontal surfaces. We have set the speed of light to $c=1$. Color online.}
    \label{fig:detection_defs}
\end{figure}

Our model of the detection process is perhaps more natural than the one at the basis of Lienert and Tumulka's theorem, as it approximates detectors on tilted surfaces through detectors on tilted hyperplanes, rather than 
on numerous small pieces of horizontal hyperplanes. On the other hand, the result of Lienert and Tumulka is stronger than ours in that it assumes the Born rule only on horizontal hyperplanes (``horizontal Born rule’’) and not on all tilted spacelike hyperplanes (``flat Born rule’’). Then again, our model allows for a somewhat simpler proof compared to that of Lienert and Tumulka, and the assumption of the Born and collapse rules on tilted hyperplanes seems natural if the workings of detectors are Lorentz invariant. Yet, our proof does not require the Lorentz invariance of the hypersurface evolution of the observed system (see also Remark~\ref{rem:cov} in Section~\ref{sec:hypersurfaceevolution}); in particular, the hypersurface evolution may involve external fields that break the Lorentz symmetry.

Other works in recent years dealing with a physical analysis of the quantum measurement process include \cite{DGZ:2004,FV18,Few19,BBFF2020}.

\subsection{Detection Process}
\label{sec:introdetectionprocess}

Our definition of the detection process is based on approximating any given Cauchy surface $\Sigma$ by spacelike surfaces $\Upsilon$ that are piecewise flat, and whose (countably many) flat pieces are 3d (non-regular) tetrahedra. We call such surfaces {\bf triangular surfaces}; see Figure~\ref{fig:Triangulation}. While the precise definition of a triangular surface will be postponed to Section~\ref{sec:definitions}, it may be useful to formulate already here a basic fact that we will prove in Section~\ref{sec:approxtriangular}:

\begin{prop}\label{prop:approx}
For every Cauchy surface $\Sigma$ in Minkowski space-time, there is a sequence $(\Upsilon_n)_{n\in\NNN}$ of triangular Cauchy surfaces that converges increasingly and uniformly to $\Sigma$.
\end{prop}

\noindent
\begin{minipage}{0.4\textwidth}
	\centering
	\scalebox{1.0}{\begin{tikzpicture}

\draw[thick,->] (0,-0.5) -- (0,4) node[anchor = south west] {$t$};
\draw[thick,->] (-0.5,0.1) -- (3.5,-0.7) node[anchor = south west] {$x^1$};
\draw[thick,->] (-0.3,-0.2) -- (1.5,1) node[anchor = west] {$x^2$};

\filldraw[fill = blue!8!white] (0.8,3.8) -- (1,3) -- (1.8,3.8) -- cycle;
\filldraw[fill = blue!16!white] (1,3) -- (2,2.8) -- (1.8,3.8) -- cycle;
\filldraw[fill = blue!14!white] (2,2.8) -- (1.8,3.8) -- (2.8,3.6) -- cycle;
\filldraw[fill = blue!10!white] (2,2.8) -- (2.8,3.6) -- (3,2.8) -- cycle;
\filldraw[fill = blue!8!white] (2.8,3.6) -- (3,2.8) -- (3.6,3.8) -- cycle;
\filldraw[fill = blue!12!white] (3,2.8) -- (3.6,3.8) -- (4,3.2) -- cycle;
\filldraw[fill = blue!14!white] (3.6,3.8) -- (4,3.2) -- (4.6,3.8) -- cycle;

\filldraw[fill = blue!24!white] (0.2,2) --(1,3) -- (1.2,1.7) -- cycle;
\filldraw[fill = blue!26!white] (1,3) -- (1.2,1.7) -- (2,2.8) -- cycle;
\filldraw[fill = blue!24!white] (1.2,1.7) -- (2,2.8) -- (2.2,1.7) -- cycle;
\filldraw[fill = blue!20!white] (2,2.8) -- (2.2,1.7) -- (3,2.8) -- cycle;
\filldraw[fill = blue!20!white] (2.2,1.7) -- (3,2.8) -- (3.4,2) -- cycle;
\filldraw[fill = blue!22!white]  (3,2.8) -- (3.4,2) -- (4,3.2) -- cycle;

\filldraw[fill = blue!4!white] (-0.6,1.3) -- (0.2,2) -- (0.2,1.4) -- cycle;
\filldraw[fill = blue!14!white] (0.2,2) -- (1.2,1.7) -- (0.2,1.4) -- cycle;
\filldraw[fill = blue!10!white] (1.2,1.7) -- (0.2,1.4) -- (1.4,1.2) -- cycle;
\filldraw[fill = blue!10!white] (1.2,1.7) -- (1.4,1.2) -- (2.2,1.7) -- cycle;
\filldraw[fill = blue!10!white] (2.2,1.7) -- (2.6,1.3) -- (3.4,2) -- cycle;
\filldraw[fill = blue!16!white] (2.6,1.3) -- (3.4,2) -- (4,1.2) -- cycle;
\filldraw[line width = 2, draw = red!50!black,  fill = blue!6!white] (1.4,1.2) -- (2.2,1.7) -- (2.6,1.3) -- cycle;

\draw[thick,red!50!black] (2.3,1.3) -- ++(0.3,-0.5) node[anchor = north west] {$ \Delta_k $};
\node[blue!50!black] at (4.2,4.4) {\Large $ \Upsilon $};

\end{tikzpicture}}
	\captionof{figure}{Part of a triangular surface $ \Upsilon $ in $1+2$ dim. Color online.}
    \label{fig:Triangulation}
\end{minipage}
\hfill
\begin{minipage}{0.54\textwidth}
	\centering
    \begin{tikzpicture}[declare function = {
	f(\x) = 0.05*((\x-3.5)^3 - 9*(\x-3.5))*(1/(1+0.1*(\x-3.5)^2)) + 2.5;}]

\def\dyA{0.15};
\def\dyB{0.3};
\def\dyC{0.6};
\def\dyD{1.5};

\def\Na{30}
\def\xmax{6}
\def\step{\xmax / \Na}

\draw[thick,->] (0,-0.5) -- (0,3) node[anchor = south west] {$ t $};
\draw[thick,->] (0,0) -- (6,0) node[anchor = south west] {$ x $};
\draw[thick,red,->] (-0.3,1) -- ++(0,1.2) ;

\draw[domain=0:6, samples = 20, smooth, line width = 2, blue!50!black] plot({\x}, {f(\x)}) node[anchor = south west] {$ \Sigma $};

\draw[domain = 0:6, samples = 4, red, mark = *, mark size = 1pt] plot({\x}, {f(\x)-\dyD}) ;
\draw[domain = 0:6, samples = 8, red, mark = *, mark size = 1pt] plot({\x}, {f(\x)-\dyC}) ;
\draw[domain = 0:6, samples = 16, red, mark = *, mark size = 1pt] plot({\x}, {f(\x)-\dyB}) ;
\draw[domain = 0:6, samples = 32, red, mark = *, mark size = 0.5pt] plot({\x}, {f(\x)-\dyA}) ;
\draw[red](5.7,1.9) -- ++(0.4,0.1) node[anchor = west] {$ \Upsilon_3 $};
\draw[red](5.7,1.6) -- ++(0.4,-0.1) node[anchor = west] {$ \Upsilon_2 $};
\draw[red](5.7,0.8) -- ++(0.8,-0.3) node[anchor = west] {$ \Upsilon_1 $};

\end{tikzpicture}
	\captionof{figure}{A sequence of triangular surfaces $ \Upsilon_n$ coverging increasingly and uniformly to $\Sigma $ in $1+1$ dim. Color online.}
    \label{fig:Cauchyconvergence_definition}
\end{minipage}

\vspace{0.5cm}

Here, ``increasing'' means that\footnote{In this paper, the ``future'' of 
a set $R$ in space-time means the \emph{causal future}, often denoted $J^+(R)$ \cite{ON:1983}, as opposed to the timelike future $I^+(R)$; note that $R\subseteq J^+(R)$; likewise for the ``past.''} $\Upsilon_{n+1}\subseteq \mathrm{future}(\Upsilon_n)$ for all $n$; see Figure~\ref{fig:Cauchyconvergence_definition}. Uniform convergence in a given Lorentz frame means that for every $\varepsilon>0$, all but finitely many $\Upsilon_n$ lie in $\{x+(s,0,0,0):x\in \Sigma,|s|<\varepsilon\}$; equivalently, since $\Sigma$ is the graph of a function $f:\RRR^3\to\RRR$ and $\Upsilon_n$ the graph of a function $f_n:\RRR^3\to\RRR$, uniform convergence $\Upsilon_n\to \Sigma$ means that $f_n$ converges uniformly to $f$. It turns out that this notion is Lorentz invariant:

\begin{prop}\label{prop:uniform}
If a sequence $(\Sigma_n)_{n\in\NNN}$ of Cauchy surfaces converges uniformly to a Cauchy surface $\Sigma$ in one Lorentz frame, then also in every 
other.
\end{prop}

Again, the proof is given in Section~\ref{sec:approxtriangular}. The following notation will be convenient: for any subset $A\subseteq \Sigma$, let
\begin{equation}
\begin{aligned}
	\emptyset(A) &:= \{ q \in \Gamma(\Sigma): q \cap A = \emptyset\} \\
	\exists(A)   &:= \{ q \in \Gamma(\Sigma): q \cap A \neq \emptyset\} \\
	\forall(A)   &:= \{ q \in \Gamma(\Sigma): q \subseteq A\}
\label{eq:setsofR}
\end{aligned}
\end{equation} 
be the sets of configurations with no, at least one, or all particles in $A$ (see Figure \ref{fig:EmptyExistsForall}). Note that $\exists(A)^c = 
\emptyset(A) = \forall(A^c)$, where $A^c$ means the complement of $A$ with respect to $\Sigma$. We also briefly write $\forall A$ for $\forall (A)$, and similarly $\exists A$ and $\emptyset A$.

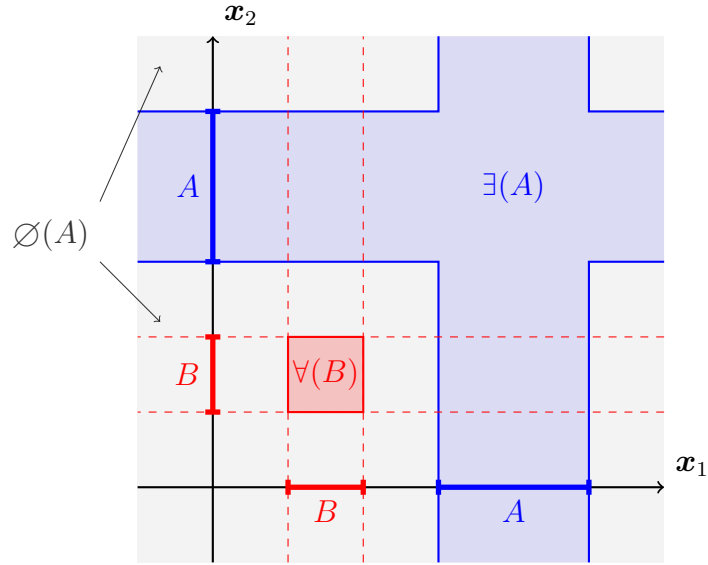
\begin{figure}[hbt]
    \centering
    \scalebox{1}{\begin{tikzpicture}
\def\Aa{(5,-1)} \def\Ab{(5,3)} \def\Ac{(6,3)}
\def\Ba{(6,5)} \def\Bb{(5,5)} \def\Bc{(5,6)}
\def\Ca{(3,6)} \def\Cb{(3,5)} \def\Cc{(-1,5)}
\def\Da{(-1,3)} \def\Db{(3,3)} \def\Dc{(3,-1)}

\def\Ea{(2,-1)} \def\Eb{(2,1)} \def\Ec{(6,1)}
\def\Fa{(6,2)} \def\Fb{(2,2)} \def\Fc{(2,6)}
\def\Ga{(1,6)} \def\Gb{(1,2)} \def\Gc{(-1,2)}
\def\Ha{(-1,1)} \def\Hb{(1,1)} \def\Hc{(1,-1)}

\fill[opacity = 0.05] (-1,-1) -- (6,-1) -- (6,6) -- (-1,6) -- cycle;

\draw[thick, ->] (-1,0) -- (6,0) node[anchor =south west] {$ \boldsymbol{x}_1 $};
\draw[thick, ->] (0,-1) -- (0,6) node[anchor =south west] {$ \boldsymbol{x}_2 $};

\draw[thick, blue] \Aa -- \Ab -- \Ac;
\draw[thick, blue] \Ba -- \Bb -- \Bc;
\draw[thick, blue] \Ca -- \Cb -- \Cc;
\draw[thick, blue] \Da -- \Db -- \Dc;
\fill[blue, opacity = 0.1]  \Aa -- \Ab -- \Ac -- \Ba -- \Bb -- \Bc -- \Ca -- \Cb -- \Cc -- \Da -- \Db -- \Dc ;
\node[blue] at (4,4) {$ \exists(A) $};
\draw[opacity = 0.8, ->] (-1.5,3) node[anchor = south east] {$ \emptyset(A) $} -- ++(0.8,-0.8) ;
\draw[opacity = 0.8, ->] (-1.5,3.8) -- ++(0.8,1.8);

\draw[blue, line width = 2] (3,0) -- node[anchor = north] {$A$} (5,0);
\draw[blue, line width = 2] (3,-0.1) -- ++(0,0.2);
\draw[blue, line width = 2] (5,-0.1) -- ++(0,0.2);

\draw[blue, line width = 2] (0,3) -- node[anchor = east] {$A$} (0,5);
\draw[blue, line width = 2] (-0.1,3) -- ++(0.2,0);
\draw[blue, line width = 2] (-0.1,5) -- ++(0.2,0);

\draw[red, dashed] \Ea -- \Fc;
\draw[red, dashed] \Hc -- \Ga;
\draw[red, dashed] \Ec -- \Ha;
\draw[red, dashed] \Fa -- \Gc;
\filldraw[red, thick, fill opacity = 0.2] \Eb -- \Fb -- \Gb -- \Hb -- cycle;
\node[red] at (1.5,1.5) {$ \forall(B) $};

\draw[red, line width = 2] (1,0) -- node[anchor = north] {$B$} (2,0);
\draw[red, line width = 2] (1,-0.1) -- ++(0,0.2);
\draw[red, line width = 2] (2,-0.1) -- ++(0,0.2);

\draw[red, line width = 2] (0,1) -- node[anchor = east] {$B$} (0,2);
\draw[red, line width = 2] (-0.1,1) -- ++(0.2,0);
\draw[red, line width = 2] (-0.1,2) -- ++(0.2,0);

\end{tikzpicture}}
    \caption{The sets $ \emptyset(\cdot), \exists(\cdot) $ and $ \forall(\cdot) $ on the 2-particle sector of configuration space, visualized. Color online.}
    \label{fig:EmptyExistsForall}
\end{figure}

We define the detection distribution on $\Sigma$ as the limit of the detection distributions on the $\Upsilon_n$, and we show in Theorem~\ref{thm:1} that this limit exists and agrees with $|\Psi_\Sigma|^2$. But to this end, 
we first need to talk about detection probabilities on triangular surfaces $\Upsilon$.

So let $\Delta_k$ be the open and disjoint tetrahedra such that
\be
\Upsilon=\bigcup_{k\in\sK} \overline{\Delta_k}
\ee
(the bar denotes closure, $\sK$ is a countably infinite index set). We want to consider a detector in a bounded region $B\subset \Upsilon$ that yields outcome 1 if there is at least one particle in $B$ and outcome 0 if there is no particle in $B$. To this end, we imagine several smaller detectors, one in each region $B_k := B\cap \Delta_k$, and set the $B$-outcome equal to 1 whenever any of the small detectors clicked. Now each region $B_k$, being a subset of $\Delta_k$, lies in some hyperplane $E_k$, and on hyperplanes we assume the Born rule and collapse rule:

\bigskip

\noindent{\bf Flat Born rule.} {\it If on the hyperplane $E$ the state vector is $\Psi_E\in\Hilbert_E$ with $\|\Psi_E\|=1$, and a detection is attempted in the region $B\subseteq E$, then the probability of outcome 1 is $\|P_E(\exists(B))\, \Psi_E\|^2$ and that of outcome 0 is $\|P_E(\exists(B)^c)\, \Psi_E\|^2$.}

\bigskip

\noindent{\bf Flat collapse rule.} {\it If the outcome is 1, then the collapsed wave function is}
\be
\Psi'_E = \frac{P_E(\exists(B))\, \Psi_E}{\|P_E(\exists(B))\, \Psi_E\|}\,, 
\ee
{\it otherwise}
\be
\Psi'_E = \frac{P_E(\exists(B)^c)\, \Psi_E}{\|P_E(\exists(B)^c)\, \Psi_E\|} \,.
\ee

\bigskip

There are two natural possibilities for defining the detection probabilities on $\Upsilon$ in terms of those on $E_k$: the sequential detection process and the parallel detection process. According to the {\bf sequential detection process}, we choose an arbitrary ordering of the set $\sK$ indexing the tetrahedra or hyperplanes and carry out, in this order, a quantum measurement in each $E_k$ representing the detection attempt in $B_k$ 
including appropriate collapse and then use the unitary evolution $U_{E_k}^{E_{k+1}}$ to evolve to the next hyperplane in the chosen order, here written as $E_{k+1}$. For the {\bf parallel detection process}, consider the projection operators $P_{E_k}(\exists(B_k))$ associated with attempted 
detection in $B_k$; we show that they, after being transferred to $\Hilbert_{\Upsilon}$ by means of $U_{E_k}^\Upsilon$, commute with each other if 
interaction locality holds, so they can be ``measured simultaneously.'' The simultaneous quantum measurement of these projections in $\Hilbert_{\Upsilon}$ provides the parallel detection process for $B\subset \Upsilon$ with outcome 1 whenever any of the quantum measurements yielded 1. It turns out that the sequential and the parallel process agree with each other 
and with the Born rule on $\Upsilon$:

\begin{prop}\label{prop:sequential1}
Fix a hypersurface evolution satisfying interaction locality (IL) (Definition \ref{def:IL}), a triangular Cauchy surface $\Upsilon$, a bounded subset $B\subset\Upsilon$, and a normalized quantum state $\Psi$, and assume the flat Born rule and the flat collapse rule. The sequential detection process in any order of the tetrahedra of $\Upsilon$ yields the same detection probability, called $\Prob^\Psi_B$; it agrees with the one given by the curved Born distribution on $\Upsilon$, which is $\|P_{\Upsilon}(\exists(B))\Psi_\Upsilon\|^2$. Moreover, the parallel detection process also yields the same detection probability.
\end{prop}

Proposition~\ref{prop:sequential1} will follow as a direct consequence of Proposition~\ref{prop:sequential3} in Section~\ref{sec:detection}.\\

Actually, for either a triangular surface $\Upsilon$ or a general Cauchy surface $\Sigma$, we want more than just to detect for a subset $B$ whether there is a particle in $B$. We want to allow the use of several detectors, each covering a region $P_1,\ldots,P_r\subset \Sigma$; the outcome of the experiment is $L=(L_1,\ldots,L_r)$ with $L_\ell=1$ if a particle gets detected in $P_\ell$ and $L_\ell=0$ otherwise. It seems physically reasonable that the region covered by a detector is bounded and has boundary of measure zero.

\begin{definition}\label{def:admissible}
An {\bf admissible partition} $\sP=(P_1,\ldots,P_r)$ of $\Sigma$ is defined by choosing finitely many subsets $P_\ell$ of $\Sigma$ that are mutually disjoint, $P_\ell\cap P_m = \emptyset$ for $\ell\neq m$, and such that each $P_\ell$ is bounded and has boundary in $\Sigma$ of (invariant) 3-volume 0. Here, the term bounded refers to the Euclidean norm on $ \RRR^4 $. We set $P_{r+1}=\Sigma \setminus (P_1 \cup \ldots \cup P_r)$ to make $(P_1,\ldots,P_{r+1})$ a partition of $\Sigma$.
\end{definition}

The idea is that there is no detector in $P_{r+1}$. Let $M_{\sP}(L)$ denote the set of configurations in $\Gamma(\Sigma)$ such that, for each $\ell=1,\ldots,r$, there is no point in $P_\ell$ if $L_\ell=0$ and at least one point in $P_\ell$ if $L_\ell=1$; that is, $M_{\sP}(L)$ {\it is the set of configurations compatible with outcome $L$.}

Now the definition of detection probabilities on a triangular surface $\Upsilon$ can straightforwardly be generalized from a bounded set $B\subset 
\Upsilon$ to an admissible partition $\sP=(P_1,\ldots,P_r)$ of $\Upsilon$ in both the sequential and the parallel sense, and we find:

\begin{prop}\label{prop:sequential2}
Fix a hypersurface evolution satisfying interaction locality, a triangular Cauchy surface $\Upsilon$, an admissible partition $\sP=(P_1,\ldots,P_r)$ of $\Upsilon$, and a normalized quantum state $\Psi$, and assume the 
flat Born rule and the flat collapse rule. The joint distribution $\Prob^\Psi_{\sP}(L)$ of $L=(L_1,\ldots,L_r)$ according to the sequential detection process in any order of the tetrahedra of $\Upsilon$ and according to the parallel detection process agree with each other and with the one given by the curved Born distribution on $\Upsilon$, which is $\|P_{\Upsilon}(M_{\sP}(L))\Psi_\Upsilon\|^2$.
\end{prop}

Proposition~\ref{prop:sequential2} can be regarded as a statement of the Born rule on triangular surfaces. It follows from Proposition~\ref{prop:sequential3}, which is proven in Section~\ref{sec:detection}.

\subsection{Main Result}
\label{sec:intromainresult}

Before we elucidate the result, let us briefly introduce some more terminology.

\begin{definition}
Let $ \Sigma, \Sigma' $ be Cauchy surfaces and $ A \subseteq \Sigma $. We 
then define the {\bf grown set} of $ A $ in $ \Sigma' $ as (see Figure~\ref{fig:GrSr})
\begin{equation}
	\Gr(A,\Sigma') = [\future(A) \cup \past(A)] \cap \Sigma'.
\label{eq:Gr}
\end{equation}
Similarly, we define the {\bf shrunk set} of $ A $ in $ \Sigma' $ as:
\begin{equation}
	\Sr(A,\Sigma') = \{x' \in \Sigma' : \Gr( \{x'\}, \Sigma) \subseteq A \}.
\label{eq:Sr}
\end{equation}
\end{definition}

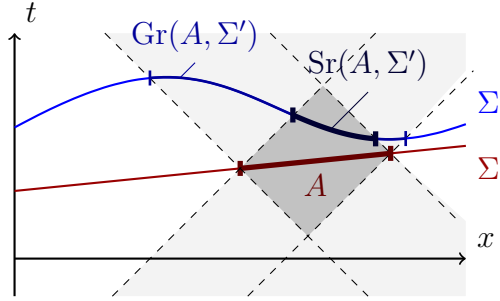
\begin{figure}[hbt]
	\centering
    \scalebox{1.0}{\begin{tikzpicture}[declare function = {
	f(\x) = 0.05*((\x-3.5)^3 - 9*(\x-3.5))*(1/(1+0.1*(\x-3.5)^2)) + 2;}]

\def\Na{30}
\def\xmax{6}
\def\step{\xmax / \Na}

\fill[opacity = 0.05] (1.2,3) -- (6,3) -- (6,2.5) -- (5,1.4) -- (6,0.5) -- (6,-0.5) -- (1.3,-0.5) -- (3,1.2) -- cycle;
\fill[opacity = 0.2] (3,1.2) -- (4.1,2.3) -- (5,1.4) -- (3.9,0.3) -- cycle ;

\draw[thick,->] (0,-0.5) -- (0,3) node[anchor = south west] {$ t $};
\draw[thick,->] (0,0) -- (6,0) node[anchor = south west] {$ x $};

\draw[thick, red!60!black] (0,0.9) -- (6,1.5) node[anchor = north west] {$ \Sigma $};
\draw[line width = 2, red!40!black] (3,1.2) --node[anchor = north] {$ A $} (5,1.4) ;
\draw[line width = 2, red!40!black] (3,1.1) -- ++ (0,0.2);
\draw[line width = 2, red!40!black] (5,1.3) -- ++ (0,0.2);

\draw[dashed] (1.3,-0.5) -- ++(3.5,3.5);
\draw[dashed] (3.1,-0.5) -- ++(3,3);
\draw[dashed] (1.2,3) -- ++(3.5,-3.5);
\draw[dashed] (3.4,3) -- ++(2.5,-2.5);

\draw[domain=0:6, samples = 20, smooth, thick, blue] plot({\x}, {f(\x)}) node[anchor = south west] {$ \Sigma' $};
\draw[domain=1.8:5.2, samples = 20, smooth, line width = 1, blue!60!black] plot({\x}, {f(\x)});
\draw[line width = 1, blue!60!black] (1.8,{f(1.8)-0.1} ) -- ++ (0,0.2);
\draw[line width = 1, blue!60!black] (5.2,{f(5.2)-0.1} ) -- ++ (0,0.2);
\draw[blue!60!black] (2.2, {f(2.2)} ) -- ++(0.2,0.2) node[anchor = south] {Gr$ (A,\Sigma') $};

\draw[domain=3.7:4.8, samples = 20, smooth, line width = 2, blue!20!black] plot({\x}, {f(\x)});
\draw[line width = 2, blue!20!black] (3.7,{f(3.7)-0.1} ) -- ++ (0,0.2);
\draw[line width = 2, blue!20!black] (4.8,{f(4.8)-0.1} ) -- ++ (0,0.2);
\draw[blue!20!black] (4.2, {f(4.2)} ) -- ++(0.5,0.5) node[anchor = south] {Sr$ (A,\Sigma') $};

\end{tikzpicture}}
	\captionof{figure}{Grown and shrunk sets of $A\subset \Sigma$. Color online.}
    \label{fig:GrSr}
\end{figure}

The following aspect of our result requires some explanation: once we have a triangular surface $\Upsilon$ approximating a given Cauchy surface $\Sigma$, and once we are given an admissible partition $\sP=(P_1,\ldots,P_r)$ on $\Sigma$, we want to approximate the sets $P_\ell\subset \Sigma$ 
by sets $B_\ell$ in $\Upsilon$. One may think of two natural possibilities of defining $B_\ell$: (i) project $P_\ell$ downwards along the direction of the $x^0$ axis of a chosen Lorentz frame; or (ii) take $B_\ell=\Sr(P_\ell,\Upsilon)$, the smallest set on $\Upsilon$ that in some sense corresponds to $P_\ell$. Our result holds in both variants; we formulate it in variant (i) (see Remark~\ref{rem:Bnell} in Section~\ref{sec:proofthm1} about (ii)). That is, choose a Lorentz frame and let
\be\label{pidef}
\pi: \RRR^4 \to \RRR^3, \qquad \pi(x^0,x^1,x^2,x^3) := (x^1,x^2,x^3)
\ee
be the projection to the space coordinates. It is known \cite[p.~417]{ON:1983} that the restriction $\pi_\Sigma$ of the projection $\pi$ to $\Sigma$ is a homeomorphism $\Sigma\to\RRR^3$; thus, $\pi_{\Sigma}^\Upsilon:=\pi^{-1}_\Upsilon \circ \pi_\Sigma$ is a homeomorphism $\Sigma\to\Upsilon$. We set
\be
B_\ell := \pi_{\Sigma}^\Upsilon (P_\ell).
\ee 
Of course, since we prove that the limiting probability distribution on $\Gamma(\Sigma)$ is given by the curved Born distribution, the limiting probability distribution is independent of the choice of Lorentz frame used 
for defining $\pi_\Sigma^\Upsilon$.

We can now state our main result.

\begin{theorem}\label{thm:1}
Let $\Sigma$ be a Cauchy surface in Minkowski space-time $\MMM$ and $(\Upsilon_n)_{n\in\NNN}$ a sequence of triangular Cauchy surfaces that converges increasingly and uniformly to $\Sigma$. Let $\sE=(\Hilbert_\circ, P_\circ, U_\circ^\circ)$ be a hypersurface evolution satisfying propagation locality and $\Psi_0\in\Hilbert_{\Sigma_0}$ with $\|\Psi_0\|=1$ for some $\Sigma_0$ in the past of $\Sigma$. Then for any admissible partition $\sP$ of $\Sigma$, $\sB_n:=\bigl(\pi_\Sigma^{\Upsilon_n}(P_1),\ldots,\pi_\Sigma^{\Upsilon_n}(P_r)\bigr)$ is an admissible partition of $\Upsilon_n$, and
\be
\lim_{n\to\infty} \Bigl\| P_{\Upsilon_n}(M_{\sB_n}(L)) \, U_{\Sigma_0}^{\Upsilon_n}\, \Psi_0\Bigr\|^2 = \Bigl\| P_\Sigma(M_{\sP}(L)) \, U_{\Sigma_0}^\Sigma \, \Psi_0\Bigr\|^2
\ee
for all $L\in\{0,1\}^r$.
\end{theorem}

Together with Proposition~\ref{prop:sequential2}, we obtain:

\begin{cor}\label{cor:main}
Assume the hypotheses of Theorem~\ref{thm:1} together with the flat Born rule, the flat collapse rule, and interaction locality.
Define the detection probabilities for $\sP$ on $\Sigma$ as the limit of the detection probabilities for $\sB_n$ on $\Upsilon_n$ and the latter through either the sequential or the parallel detection process. Then the detection probabilities for $\sP$ on $\Sigma$ are given by the curved Born rule, 
$\bigl\| P_\Sigma(M_{\sP}(L)) \, \Psi_\Sigma \bigr\|^2$ for all $L\in\{0,1\}^r$.
\end{cor}

The proof of Theorem~\ref{thm:1} (see Section~\ref{sec:proofthm1}) makes no special use of dimension $3+1$ and applies equally in dimension $d+1$ for any $d\in\NNN$; tetrahedra then need to be replaced by $d$-dimensional simplices.

\paragraph{Remarks.}

\begin{enumerate}
\setcounter{enumi}{\theremarks}
\item {\it Shrunk set $\Sr(A,\Sigma')$.}\label{rem:shrunkset} Definition \eqref{eq:Sr} is equivalent to saying that the shrunk set is the intersection of $\Sigma'$ and the domain of dependence of $\Sigma$

\item {\it Uniqueness of the measure on $\Gamma(\Sigma)$.}\label{rem:unique} It was shown in Proposition 3 in Section 6 of \cite{LT:2020} that
\emph{if two probability measures $\mu,\mu'$ on $\Gamma(\Sigma)$ agree on 
all detection outcomes, $\mu(M_{\sP}(L))=\mu'(M_{\sP}(L))$ for every $L\in\{0,1\}^r$ and every admissible partition $\sP$ of $\Sigma$, then $\mu=\mu'$.} Thus, the whole $|\Psi_\Sigma|^2$ distribution is uniquely determined by the detection probabilities. 

In fact, a probability measure $\mu$ on $\Gamma(\Sigma)$ is already uniquely determined by the values $\mu(\emptyset(A))$, where $A$ runs through those subsets of $\Sigma$ whose projection $\pi(A)$ to $\RRR^3$ is a union of finitely many open balls (see the proof of Proposition 3 in \cite{LT:2020}). This fact might suggest that, in order to prove the curved Born rule, it would have been sufficient to prove the statement of Theorem~\ref{thm:1} only for a single detector (i.e., for partitions with $r=1$ consisting of $P_1=A$ and $P_{r+1}=\Sigma \setminus A$) in a region $A$ 
of the type described. However, we prove the stronger statement for arbitrary $r$ because it is not obvious that the detection probabilities for arbitrary $r$ fit together to form a measure on $\Gamma(\Sigma)$ (in other 
words, that detection probabilities for $r>1$ will agree with the Born distribution, given that detection probabilities for $r=1$ do).

\item {\it Curved collapse rule.}\label{rem:collapse} One can also consider a {\bf curved collapse rule}: {\it Suppose that $r$ detectors are placed along $\Sigma$, that each detector (say the $\ell$-th) only measures whether there is a particle in the region $P_\ell$, where $\sP=(P_1,\ldots,P_r)$ is an admissible partition, and that each detector acts immediately (i.e., is infinitely fast). If the outcome was $L=(L_1,\ldots,L_r)\in\{0,1\}^r$, then the wave function immediately after detection is the collapsed wave function}
\be\label{collapse}
\Psi'_\Sigma = \frac{P_\Sigma(M_{\sP}(L)) \, \Psi_\Sigma}{\|P_\Sigma(M_{\sP}(L)) \, \Psi_\Sigma\|}\,.
\ee

There is a sense in which the curved collapse rule also follows from our result and a sense in which it does not. To begin with the latter, our justification of the Born rule on triangular surfaces was based on the idea 
that on each tetrahedron $\Delta_k$, we apply a detector to $B_{k\ell}=\Delta_k\cap B_\ell$ and deduce from the outcomes whether a particle has been detected anywhere in $B_\ell$. This detection process measures more than whether there is a particle in $B_\ell$, as it also measures which of the $B_{k\ell}$ contain particles; as a consequence, this detection process would collapse $\Psi$ more narrowly than \eqref{collapse}. 

However, if we assume that on triangular surfaces $\Upsilon$ we can have detectors that \emph{only} measure whether there is a particle in $B_\ell$ for an admissible partition $\sB=(B_1,\ldots,B_r)$, so that the collapse rule \eqref{collapse} holds upon replacing $\Sigma \to \Upsilon$ and $\sP\to\sB$, then sufficient approximation of an arbitrary Cauchy surface 
$\Sigma$ by triangular surfaces leads to a collapsed wave function arbitrarily close to \eqref{collapse}. Indeed, we have that (see Section~\ref{sec:proofthm1} for the proof)

\begin{cor}\label{cor:strong}
Under the hypotheses of Theorem~\ref{thm:1},
\be
U_{\Upsilon_n}^\Sigma \, P_{\Upsilon_n}(M_{\sB_n}(L)) \, U_\Sigma^{\Upsilon_n} \xrightarrow{n\to\infty} P_\Sigma(M_{\sP}(L)) ~~\text{strongly.}
\ee
\end{cor}

\item\label{rem:MeasurementProcess} {\it Other observables.} As the curved Born rule shows, the PVM $P_\Sigma$ can be regarded as the totality of position observables on $\Sigma$. What about other observables? In a sense, all other observables are indirectly determined by the position observable. As Bell \cite[p.~166]{Bell87} wrote:
\begin{quote}
[I]n physics, the only observations we must consider are position observation, if only the positions of instrument pointers. [\ldots] If you make axioms, rather than definitions and theorems, about the `measurements' of 
anything else, then you commit redundancy and risk inconsistency.
\end{quote}
A detailed description of how self-adjoint obervables arise from the Hamiltonian of an experiment, the quantum state of the measuring apparatus, and the position observable (of its pointer), can be found in \cite[Sec.~2.7]{DGZ:2004}. A conclusion we draw is that specifying a quantum theory's hypersurface evolution is an informationally complete description. 

As another conclusion, the PVM $P_\Sigma$ serves not only for representing detectors. When we want to argue that certain experiments are quantum measurements of certain observables, we may use it to link the quantum state with macro-configurations (say, of pointer positions), and in fact to obtain probabilities for pointer positions.

A related but quite different question is how the algebras of local operators common in algebraic QFT (such as smeared field operator algebras or Weyl algebras) are related to $P_\Sigma$. It would be a topic of interest 
for future work to make this relation explicit.

Coming back to the Bell quote, one may also note that for the same reason, making the curved Born rule an axiom in addition to the flat Born rule means to commit redundancy and to risk inconsistency. That is why we have 
made the curved Born rule a theorem.

Of course, we have still committed a little bit of the redundancy that Bell talked about by assuming the Born and collapse rules on all spacelike hypersurfaces while it suffices to assume them on horizontal hypersurfaces \cite{LT:2020}.

\item {\it Objections.} Some authors \cite{TV99} have criticized the very 
idea of evolving states from one Cauchy surface to another on the grounds 
that such an evolution cannot be unitarily implemented for the free second-quantized scalar Klein-Gordon field. It seems to us that these difficulties do not invalidate the approach but stem from analogous difficulties with 1-particle Klein-Gordon wave functions, which are known to lack a covariantly-defined timelike probability current 4-vector field that could be used for defining a Lorentz-invariant inner product that makes the time evolution unitary (e.g., \cite{schweber:1961}). In contrast, a hypersurface evolution according to our definition can indeed be defined for the free second-quantized Dirac equation allowing negative energies \cite{Dirac64,Dim82,DM:2016,LT:2020}. 
Other results (\cite[Sec.~1.8]{Thaller}, \cite{Heg:1998,HC:2002}) may raise doubts about propagation locality; on the other hand, these results presuppose positive energy, which we do not require here; moreover, violations of propagation locality would seem to allow for superluminal signaling. Be that as it may, we simply assume here a propagation-local hypersurface evolution as given; further developments of this notion can be of interest for future works.

\item {\it Evolution Between Hyperplanes.} Following \cite[Sec.~8]{LT:2020}, we conjecture that a hypersurface evolution $\sE$ satisfying interaction locality and propagation locality is uniquely determined up to unitary equivalence \cite[Sec.~3.2 Rem.~14]{LT:2020} by its restriction to hyperplanes. We conjecture further that a hypersurface evolution that is in addition Poincar\'e covariant (see Remark~\ref{rem:cov} in Section~\ref{sec:hypersurfaceevolution}) is uniquely determined by its restriction to horizontal hyperplanes $\{x^0=\mathrm{const.}\}$. While we do not have a proof of these statements, a related statement follows from our results: 

{\it Suppose two hypersurface evolutions $ \sE = (\Hilbert_\circ, P_\circ, U_\circ^\circ) $ and $ \tilde{\sE} = (\Hilbert_\circ, P_\circ, \tilde{U}_\circ^\circ) $ use the same Hilbert spaces and PVMs but potentially different 
evolution operators; suppose further that the evolution operators agree on hyperplanes, $U_E^{E'}=\tilde{U}_E^{E'}$ for all spacelike hyperplanes $E,E'$; finally, suppose that both $\sE$ and $\tilde{\sE}$ satisfy interaction locality and propagation locality. Then they yield the same Born distribution on every Cauchy surface $\Sigma$, i.e., for every $\Psi_0\in\Hilbert_{E_0}$ on $E_0=\{x^0=0\}$ and every $S\subseteq \Gamma(\Sigma)$,}
\be\label{EEtildeagree}
\|P_\Sigma(S)\, U_{E_0}^\Sigma\, \Psi_0\|^2 =
\|P_\Sigma(S)\, \tilde{U}_{E_0}^\Sigma\, \Psi_0\|^2 \,.
\ee

Indeed, by Remark~\ref{rem:unique}, \eqref{EEtildeagree} holds for all $S\subseteq \Sigma$ if it holds for all $M_{\sP}(L)$ for all admissible partitions $\sP$ of $\Sigma$. By Theorem~\ref{thm:1}, both sides can be expressed as the limits of detection probabilities on triangular surfaces. Those in turn can be expressed, using the sequential detection process, in terms of $U_E^{E'}$ respectively $\tilde{U}_E^{E'}$ only for hyperplanes $E,E'$, so they are equal.
\end{enumerate}	
\setcounter{remarks}{\theenumi}

\section{Definitions}
\label{sec:definitions}

\subsection{Geometric Notions}

We now begin the more technical part of this paper. 
We consider flat Minkowski spacetime $ \MMM $ in $ 3+1 $ dimensions with metric tensor $ \eta_{\mu \nu} = \diag(1,-1,-1,-1) $. Spacetime points are denoted by $ x = x^\mu = (x^0,\bx) = (x^0,x^1,x^2,x^3) $, the Minkowski square is denoted by $ x^2 = x^\mu x_\mu $, Cauchy surfaces are denoted by $ \Sigma \subset \MMM $. For piecewise 
flat Cauchy surfaces, we reserve the notation $ \Upsilon \subset \MMM $, for flat Cauchy surfaces (spacelike 3-planes), the notation $E\subset \MMM$; $E_0 = \{x^\mu: x^0 = 0\} \cong \RRR^3 $ is the time-zero hyperplane. For a topological space $ X $, we will denote by $\sB(X) $ the corresponding Borel $ \sigma $-algebra. The topology on $ \Sigma $ is that induced by the Euclidean $ \RRR^4 $-norm on $ \MMM $. Restricting the projection $ \pi $ as in \eqref{pidef} to $ \Sigma $, we obtain a homeomorphism $\pi_\Sigma= \pi|_\Sigma : \Sigma \to \RRR^3 $, which can be used to identify $\sB(\Sigma) $ with $\sB(\RRR^3) $: For $ R \subseteq \Sigma $, we have that $ R \in\sB(\Sigma) \; \Leftrightarrow \; \pi(R) \in\sB(\RRR^3) $. By Rademacher's theorem, $ \Sigma $ possesses a tangent plane almost everywhere  \cite[Sec.~3]{LT:2020}. If a tangent plane exists at 
$ x \in \Sigma $, the pullback of $ \eta_{\mu \nu} $ under the embedding $ \Sigma \hookrightarrow \MMM $ is either degenerate or a Riemann 3-metric. This metric can be 
used to define a volume measure $ \mu_\Sigma $ on $ (\Sigma,\mathscr{B}(\Sigma)) $, as well as a volume measure\footnote{One of us claimed in \cite{LT:2020} that the null sets of $\mu_\Sigma$, when projected to $\RRR^3$ with $\pi$, are exactly the null sets of the Lebesgue measure in $\RRR^3$; this is equivalent to saying that the set of points of $\Sigma$ with a lightlike tangent, when projected to $\RRR^3$, is a null set. While we conjecture that this is true, we do not see how to prove it. The statement is neither used in \cite{LT:2020} nor here.} $\mu_{\Gamma(\Sigma)}$ on $(\Gamma(\Sigma),\mathscr{B}(\Gamma(\Sigma)))$. In the configuration space $\Gamma(\Sigma)$, we denote the  $ n $-particle sector by
\begin{equation}
	\Gamma_n(\Sigma) := \{q \subseteq \Sigma : \#q = n \} \subset \Gamma(\Sigma).
\end{equation}
Note that for disjoint sets $ A \cap B = \emptyset $, we have
\begin{equation}
	\Gamma(A \cup B) \cong \Gamma(A) \times \Gamma(B)
\end{equation}
with bijective identification map $ q \mapsto (q \cap A, q \cap B) $.

\begin{definition}\label{def:triangular}
A {\bf triangular surface} is a Cauchy surface $\Upsilon \subset \MMM$ such that
\be
\Upsilon = \bigcup_{k\in\sK} \overline{\Delta_k}\,,
\ee
where $\sK$ is a countably infinite index set, each $\Delta_k$ is a 3-open, non-degenerate, spacelike tetrahedron (i.e., the non-empty 3-interior of the convex hull of $3+1$ points that are mutually spacelike), the $\Delta_k$ are mutually disjoint ($\Delta_{k_1} \cap \Delta_{k_2} = \emptyset$ for $k_1\neq k_2$), and every bounded region $B\subset \Upsilon$ intersects only finitely many $\Delta_k$. 
\end{definition}

\subsection{Hypersurface Evolution}
\label{sec:hypersurfaceevolution}

\begin{definition}\label{def:hypersurfaceevolution}
A \textbf{hypersurface evolution} $\mathscr{E} = ( \Hilbert_\circ, P_\circ, U_{\circ}^{\circ})$  is a collection of
\begin{enumerate}
	\item Hilbert spaces $\Hilbert_\Sigma$ for every Cauchy surface $\Sigma$, equipped with
	\item a PVM $P_\Sigma : \sB(\Gamma(\Sigma)) \rightarrow {\rm Proj}(\Hilbert_\Sigma)$, the set of projections in $\Hilbert_\Sigma$,
	\item unitary isomorphisms $U_{\Sigma}^{\Sigma'}:\Hilbert_\Sigma \to \Hilbert_{\Sigma'}$ (``evolution''), and
	\item a factorization mapping for every $A\subseteq \Sigma$, i.e., with the abbreviation
	\be
	\Hilbert_{\Sigma,A} := \range P_\Sigma(\forall(A)),
	\ee
	(where $\range$ denotes the range), a unitary isomorphism $T_{\Sigma, A}:\Hilbert_\Sigma \to \Hilbert_{\Sigma,A} \otimes \Hilbert_{\Sigma,\Sigma\setminus A}$ (``translation'')
\end{enumerate}
with the following properties:
\begin{enumerate}
	\item[(0)] $U_\Sigma^\Sigma = I_\Sigma$ and $U_{\Sigma'}^{\Sigma''} U_{\Sigma}^{\Sigma'} = U_{\Sigma}^{\Sigma''}$ for all Cauchy surfaces $\Sigma, \Sigma', \Sigma''$.
	\item[(i)] For every $S\subset \Gamma(\Sigma)$ with $\mu_{\Gamma(\Sigma)}(S) = 0$, also $P_\Sigma(S) = 0$.
	\item[(ii)] For every $\Sigma$,  $\dim \range P_\Sigma(\emptyset(\Sigma))=1$. That is, up to a phase, there is a unique \textit{vacuum state} $| \emptyset(\Sigma) \rangle \in \range P_\Sigma(\emptyset(\Sigma))$ with $\bigl\|  |\emptyset(\Sigma) \rangle \bigr\| = 1$.
	\item[(iii)] $T_{\Sigma, \Sigma\setminus A} = \Pi T_{\Sigma, A}$ with $\Pi$ the permutation of two tensor factors
	\item[(iv)] Factorization of the PVM:\footnote{\label{fn:PSigmaA}Note that $P_\Sigma$, restricted to subsets of $\forall (A)$, maps $\Hilbert_{\Sigma, A}$ to itself and in fact defines a PVM on $\Hilbert_{\Sigma, A}$.} 
	For all $A,B \subseteq \Sigma$,
	\be\label{Pfactor1}
	P_\Sigma(\forall (B))= T_{\Sigma, A}^{-1} 
	\bigl[P_\Sigma(\forall(A\cap B)) \otimes 
	P_\Sigma(\forall (A^c\cap B))\bigr] T_{\Sigma, A} \,.
	\ee
\end{enumerate}
\end{definition}

This definition is equivalent to the one given in \cite{LT:2020} but formulated in a more detailed way, as the isomorphisms $T$ were previously not made explicit. We will often  follow \cite{LT:2020} and not make the isomorphism $T$ explicit; that is, instead of saying ``the given unitary isomorphism $T_{\Sigma, A}$ maps $\Hilbert_\Sigma$ to $\Hilbert_{\Sigma,A} \otimes \Hilbert_{\Sigma,\Sigma\setminus A}$,'' we simply say ``$\Hilbert_\Sigma=\Hilbert_{\Sigma,A} \otimes \Hilbert_{\Sigma,\Sigma\setminus A}$.'' Likewise, instead of \eqref{Pfactor1}, we simply write 
\be\label{Pfactor2}
P_\Sigma(\forall (B))= P_A(\forall(A\cap B)) \otimes P_{A^c}(\forall (A^c\cap B))\,,
\ee
where $P_A$ means the restriction of $P_\Sigma$ to subsets of $\forall(A)$ as in Footnote~\ref{fn:PSigmaA}.

\paragraph{Remarks.}

\begin{enumerate}
\setcounter{enumi}{\theremarks}
\item {\it Uniqueness of the vacuum state.} Actually, our Propositions and the Theorem do not make use of property (ii), the uniqueness of the vacuum state. The reason we make it part of the definition of $\sE$ is that it is part of the concept of hypersurface evolution as introduced in~\cite{LT:2020}.\\

\item {\it $P_\Sigma$ factorizes.} From \eqref{Pfactor1} or \eqref{Pfactor2} it follows that $P_\Sigma$ factorizes not just for all-sets (i.e., sets of the form $\forall(B)$) but for all product sets in configuration space: for all $A\subseteq \Sigma$, $S_A\subseteq \forall A$, and 
	$S_{A^c}\subseteq \forall(\Sigma \setminus A)$,
	\be\label{Pfactor}
	P_\Sigma(S_A\times S_{A^c})= P_A(S_A) \otimes 
	P_{A^c}(S_{A^c})	
	\ee
with $S_A \times S_{A^c}$ understood as a subset of $\Gamma(\Sigma)$. That is because, first, $\forall B = \forall(A\cap B) \times \forall(A^c \cap B)$, second, the all-sets $\forall C$ form a $\cap$-stable generator of $\sB(\Gamma(\Sigma))$, and third, it is a standard theorem in probability theory that measures (and hence also PVMs) agreeing on a $\cap$-stable generator of a $\sigma$-algebra agree on the whole $\sigma$-algebra; so, roughly speaking, relations true for all all-sets are true for all sets. Relation \eqref{Pfactor} is exactly the definition of the tensor product of two POVMs, so it can equivalently be expressed as
\be\label{Pfactor3}
P_\Sigma = P_A \otimes P_{A^c}\,.
\ee

\item\label{rem:Pfactorfinite} {\it Splitting into more than two regions.} The restriction $T_{\Sigma, B, A}$ of $T_{\Sigma, A}$ to $\Hilbert_{\Sigma, B}$ maps $\Hilbert_{\Sigma, B}$ unitarily to $\Hilbert_{\Sigma,A\cap B}\otimes \Hilbert_{\Sigma,A^c\cap B}$. Moreover, \eqref{Pfactor1} for $A\subseteq B$ yields that $P$ factorizes also in $B$, i.e., for every $A\subseteq B\subseteq \Sigma$, $S_A\subseteq \forall A$, and $S_{B \setminus A}\subseteq \forall(B \setminus A)$,
	\be
	P_\Sigma(S_A\times S_{B \setminus A})= T_{\Sigma, B, A}^{-1} 
	\bigl[P_\Sigma(S_A) \otimes P_\Sigma(S_{B \setminus A})\bigr] 
	T_{\Sigma, B, A}
	\ee
with $S_A \times S_{B\setminus A}$ understood as a subset of $\forall B$. 
Furthermore, it follows that $T_{\Sigma, B, B\setminus A} = \Pi T_{\Sigma, B, A}$, and that an associative law holds for the $T_{\Sigma,B,A}$: For any partition $A_1,A_2,A_3$ of $B\subseteq \Sigma$,
	\be\label{Tassociative}
	\bigl(I_{\Sigma, A_1} \otimes T_{\Sigma, A_2\cup A_3, A_2} \bigr)T_{\Sigma, B, A_1}
	= \bigl(T_{\Sigma,A_1\cup A_2,A_1}\otimes I_{\Sigma, A_3} \bigr) T_{\Sigma, B,A_1\cup A_2}\,.
	\ee
Hence, the Hilbert spaces and PVMs factorize also for finite partitions. The upshot is that it is OK to identify
\begin{align}
\Hilbert_{\Sigma}&=\bigotimes_i \Hilbert_{\Sigma, A_i}~\text{and}\\
P_\Sigma&=\bigotimes_i P_{A_i} \label{Pfactorfinite}
\end{align}
for any finite partition $\Sigma=\bigcup_i A_i$.

\item\label{rem:exampleE} {\it Examples for hypersurface evolutions $\sE$.}
Some examples for hypersurface evolutions can be found in \cite{LT:2020}. As described there in Remark 15 and Section 4.1, the simplest example is provided by the non-interacting Dirac field without a Dirac sea, which also satisfies (IL) and (PL) as defined below. Further examples are provided by Tomonaga-Schwinger equations and multi-time wave functions (whose $n$-particle sectors are functions of $n$ space-time points, rather than $n$ space points \cite{lienertpetrattumulka}); explicit models include the emission-absorption model of \cite{pt:2013c} and the rigorous model with contact interaction of \cite{lienert:2015a,LN:2015}. Given an evolution law for multi-time wave functions $\phi$, $U_\Sigma^{\Sigma'}$ can be defined by $U_\Sigma^{\Sigma'}:\phi|_\Sigma \mapsto \phi|_{\Sigma'}$; of course, one still has to check that this $U_\Sigma^{\Sigma'}$ is indeed unitary. In fact, multi-time wave functions have provided a major motivation for considering the curved Born rule.\\

\end{enumerate}
\setcounter{remarks}{\theenumi}

\subsection{Locality Properties}

\begin{definition}
$\sE$ is {\bf propagation local (PL)} if and only if
\be\label{PL}
U^{\Sigma'}_\Sigma \, P_\Sigma (\forall A) \, U^\Sigma_{\Sigma'} 
\leq P_{\Sigma'}(\forall \Gr(A,\Sigma'))
\ee
for all Cauchy surfaces $\Sigma,\Sigma'$ and all $A\subseteq \Sigma$.
\end{definition}

\bigskip

Here, $R\leq S$ means that $S-R$ is a positive operator; if $R$ and $S$ are projections, then $R\leq S$ is equivalent to $\range R \subseteq \range S$. In words, (PL) means that if $ \Psi_\Sigma $ is concentrated in $ A \subseteq \Sigma $, i.e., $ \Psi_\Sigma \in \Hilbert_{\Sigma, A} $, then $ \Psi_{\Sigma'} = U_{\Sigma}^{\Sigma'} \Psi_{\Sigma} $ is concentrated in $\Gr(A,\Sigma') $. Also this definition is equivalent to the one given in \cite{LT:2020}.

Also the definition of interaction locality was already given in \cite{LT:2020} but will be formulated here in a more detailed way. We begin with a summary of the condition: First, in a region $A$ where $\Sigma$ and $\Sigma'$ overlap (see Figure~\ref{fig:IL}), $\Hilbert_{\Sigma, A}$ and $\Hilbert_{\Sigma', A}$ can be identified.
The identification fits together with $P$ and $T$.
Second, the time evolution from $\Sigma\setminus A$ to $\Sigma'\setminus A$ (see Figure~\ref{fig:IL}) is given by a unitary isomorphism $V_{\Sigma\setminus A}^{\Sigma'\setminus A}:\Hilbert_{\Sigma\setminus A} \to \Hilbert_{\Sigma' \setminus A}$, the ``local evolution'' replacing $U_\Sigma^{\Sigma'}$. The fact that one can evolve from $\Sigma\setminus A$ to $\Sigma'\setminus A$ means in particular that this evolution does not depend on 
the state in $A$, that is, there is no interaction term in the evolution that would couple $\Sigma\setminus A$ to $A$. Finally, we require that $V_{\Sigma\setminus A}^{\Sigma'\setminus A}$ does not change when we deform 
$A$ while keeping it spacelike from $\Sigma\setminus A$.

\bigskip

\begin{definition}\label{def:IL}
$\sE$ is {\bf interaction local (IL)} if it is equipped in addition with, 
for all Cauchy surfaces $\Sigma,\Sigma'$ and $A\subseteq \Sigma\cap\Sigma'$, a unitary isomorphism $J_{A, \Sigma}^{\Sigma'}: \Hilbert_{\Sigma, A} \to 
\Hilbert_{\Sigma', A}$ (``identification'') such that
\begin{align}
J^{\Sigma''}_{A, \Sigma'} \, J^{\Sigma'}_{A, \Sigma} &= J^{\Sigma''}_{A, \Sigma}
\text{ whenever }A\subseteq \Sigma\cap \Sigma' \cap \Sigma''\,,
\label{ILJ}\\
J^{\Sigma'}_{B, \Sigma} &= J^{\Sigma'}_{A, \Sigma}\Big|_{\Hilbert_{\Sigma, B}}~~~\text{for }B\subseteq A\subseteq \Sigma\cap \Sigma'\,,\\
(J^{\Sigma'}_{A, \Sigma})^{-1}\, P_{\Sigma'}(\forall B) \, J^{\Sigma'}_{A, \Sigma} &= P_\Sigma(\forall B)~~\text{for }B\subseteq A\,,\label{ILJP}\\[2mm]
T_{\Sigma', A} \, U^{\Sigma'}_\Sigma \, T_{\Sigma, A}^{-1} &= J^{\Sigma'}_{A, \Sigma} \otimes V^{\Sigma'}_{\Sigma\setminus A,\Sigma}\label{ILT}
\end{align}
with some unitary isomorphism
$V^{\Sigma'}_{\Sigma\setminus A,\Sigma} : 
\Hilbert_{\Sigma,\Sigma\setminus A} \to \Hilbert_{\Sigma',\Sigma'\setminus A}$
such that for all $\tilde\Sigma\supseteq (\Sigma\setminus A)$, setting $\tilde A:= \tilde\Sigma\setminus(\Sigma\setminus A)$ and $\tilde{\Sigma}' := \tilde A \cup (\Sigma'\setminus A)$,
\be\label{VJVJ}
V^{\tilde{\Sigma}'}_{\Sigma \setminus A, \tilde\Sigma} = J^{\tilde{\Sigma}'}_{\Sigma'\setminus A,\Sigma'} \, V^{\Sigma'}_{\Sigma \setminus A, \Sigma} \, J^\Sigma_{\Sigma \setminus A, \tilde\Sigma}\,.
\ee
\end{definition}

Henceforth, we will not mention the $ J $-operators explicitly any more and following \cite{LT:2020}, we will simply write
\be\label{HilbertA}
\Hilbert_{\Sigma, A} = \Hilbert_{\Sigma', A} =: \Hilbert_A\,.
\ee
Further, we will write $ V_{\Sigma \setminus A}^{\Sigma' \setminus A} $ in place of $ V_{\Sigma \setminus A, \Sigma}^{\Sigma'} $, which is compatible with the Hilbert space identification.\\

\begin{figure}
	\centering
	\scalebox{1.0}{\begin{tikzpicture}

\draw (0,0.5) .. controls (1,0) and (1.5,0) .. (2,0) .. controls (2.5,0) and (4,0) .. (5,0.5) .. controls (5.5,0.7) and (5.5,0.7) .. (6,0.7);
\draw (2,0) .. controls (2.5,0) and (3,1) .. (3.5,1) .. controls (4,1) and (4.5,0.25) .. (5,0.5);
\draw (3.5,0.1) -- ++(0.7,-0.7) node[anchor = north west] {$ \Sigma \setminus A $};
\draw (3.5,1) -- ++(0.7,0.7) node[anchor = south west] {$ \Sigma' \setminus A $};

\draw[line width = 3, blue] (0,0.5) .. controls (1,0) and (1.5,0) .. node[anchor = south] {$A$} (2,0);
\draw[line width = 3, blue] (5,0.5) .. controls (5.5,0.7) and (5.5,0.7) .. (6,0.7);
\draw[line width = 3, blue] (2,-0.1) -- ++(0,0.2);
\draw[line width = 3, blue] (5,0.4) -- ++(0,0.2);

\draw[dotted,->,red!50!blue] (3,0.2) -- ++(0,0.3);
\draw[dotted,->,red!50!blue] (3.5,0.2) -- ++(0,0.6);
\draw[dotted,->,red!50!blue] (4,0.3) -- ++(0,0.3);

\draw[<-,blue] (1,0) -- ++(0.5,-1.2) node[anchor = north] {$ \Psi $ remains invariant};

\end{tikzpicture}}
	\captionof{figure}{Depiction of interaction locality (IL). Color online.}
    \label{fig:IL}
\end{figure}
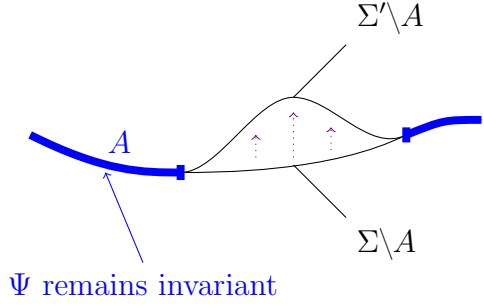

\paragraph{Remarks.} 

\begin{enumerate}
\setcounter{enumi}{\theremarks}

\item {\it Other notions of locality.} There are several inequivalent (though not unrelated) concepts of locality; they often play important roles in selecting time evolution laws (e.g., \cite{Haag96,Sch95}).

In the Wightman axioms (e.g., \cite[p.~65]{reedsimon2}), a locality condition appears that is different from both (IL) and (PL), viz., (anti-)commutation of field operators at spacelike separation. It seems clear that Wightman's locality is closely related to (IL) and (PL), and it would be of interest to study this relation in detail in a future work.

Another different locality condition is often called Einstein locality or 
Bell locality or just locality. It implies (IL) and (PL) but is not implied by (IL) and (PL) together; it asserts that there are no influences between events in spacelike separated regions; that may sound similar to (IL), but it is not. In fact, Bell's theorem \cite{Bell64,GNTZ} shows that Bell locality is violated, whereas (IL) seems to be valid in our universe.

\item {\it Consistency condition.}\label{rem:consistencycondition} It is known that multi-time equations require a consistency condition (e.g., \cite[Chap.~2]{lienertpetrattumulka}). We note here that neither (IL) nor (PL) follow from the consistency condition alone. Indeed, examples of (artificial) multi-time equations with an instantaneous interaction (violating (IL)) that leaves the multi-time equations consistent were given in Lemma 2.5 of \cite{DN:2016}, while the non-interacting multi-time equations with Schr\"odinger Hamiltonians $-\Delta_j$ for each particle $j$ provide an example of consistent multi-time equations violating (PL).

\item {\it Poincar\'e covariance.}\label{rem:cov} While the flat Born rule is inspired by the thought that the full theory should be covariant under Poincar\'e transformations (i.e., Lorentz transformation and space-time translations), we do not assume covariance of the hypersurface evolution. To make this point, it may be helpful to say explicitly what it would mean for a hypersurface evolution $(\Hilbert_\circ,U_\circ^\circ,P_\circ)$ to be Poincar\'e covariant: It would mean that {\it for every proper\footnote{A {\bf proper} Poincar\'e transformation is one that reflects neither space nor time; the set of proper Poincar\'e transformations is often denoted by $\cP_+^{\uparrow}$.} Poincar\'e transformation $g$ and every Cauchy surface $\Sigma$ there is a unitary isomorphism $S_{g,\Sigma}: \Hilbert_\Sigma \to \Hilbert_{g\Sigma}$ (thought of as just Poincar\'e transforming the wave function without evolving it) such that}
\be\label{SS}
S_{\mathrm{id},\Sigma}=I_\Sigma\,,~~~S_{h,g\Sigma} \, S_{g,\Sigma} = S_{hg,\Sigma}
\ee
\be\label{US}
U^{g\Sigma'}_{g\Sigma} = S_{g,\Sigma'} \, U^{\Sigma'}_\Sigma \, S_{g,\Sigma}^{-1}
\ee
\be\label{PS}
P_{g\Sigma}(\forall(gA))=S_{g,\Sigma}\, P_\Sigma(\forall A) \, S^{-1}_{g,\Sigma}
\ee
\be\label{TS}
T_{g\Sigma,gA} \, S_{g,\Sigma} \, T_{\Sigma, A}^{-1} = S_{g,\Sigma}|_{\sH_{\Sigma,A}} \otimes S_{g,\Sigma}|_{\sH_{\Sigma,\Sigma \setminus A}}
\ee
{\it with $T_{\Sigma, A}$ as in Definition~\ref{def:hypersurfaceevolution} 
item 4.}

The representation $U(g)$ of the proper Poincar\'e group on $\Hilbert_{E_0}$ ($E_0=\{x^0=0\}$) that features (e.g.) in the Wightman axioms (e.g., \cite[p.~65]{reedsimon2}) corresponds to
\be\label{Ugdef}
U(g) = U^{E_0}_{gE_0} \, S_{g,E_0}\,,
\ee
that is, to using the Poincar\'e transformation $g$ to shift $\Psi$ from $E_0$ and subsequently using the time evolution to bring the state vector 
back to $E_0$.
\end{enumerate}
\setcounter{remarks}{\theenumi}

\section{Detection Process on Triangular Surfaces}
\label{sec:detection}

We now give the detailed definitions of the sequential and parallel detection processes and prove Propositions~\ref{prop:sequential1} and \ref{prop:sequential2}. 

To begin with, consider an admissible partition $\sP=(P_1,\ldots,P_r)$ of a Cauchy surface $ \Sigma $ and a vector $L = (L_1,\ldots,L_r) \in \{0,1\}^r$. Actually, in this section we will not make use of the assumption in Definition~\ref{def:admissible} that the boundaries $\partial P_\ell$ are null sets, an assumption we need for Theorem~\ref{thm:1}.

The set of configurations in $ \Gamma(\Sigma) $ compatible with the single outcome $L_\ell$ at an attempted detection in $P_\ell$ is
\begin{equation}
	M_{\ell\Sigma}(L_\ell) := \begin{cases} \exists(P_\ell) \quad &\text{if } L_\ell = 1 \\ \emptyset(P_\ell)\quad &\text{if } L_\ell = 0 \end{cases}.
\end{equation}
The set of configurations compatible with the measurement outcome vector $ L $ when detection is attempted in $P_1,\ldots,P_r$ is
\begin{equation}\label{MsPdef}
	M_\sP(L) := \bigcap_{\ell=1}^r  M_{\ell\Sigma}(L_\ell)\,.
\end{equation}

Now consider a triangular surface $\Upsilon=\bigcup_{k\in\sK} \overline{\Delta_k}$ and an admissible partition $\sB=(B_1,\ldots,B_r)$ of $\Upsilon$. For either the sequential or the parallel detection process on $\Upsilon$, we imagine a small detector checking for particles in each
\be\label{Bkelldef}
B_{k\ell}:= \Delta_k\cap B_\ell 
\ee
with outcome $s_{k\ell}=1$ if a particle was found and $s_{k\ell}=0$ otherwise.\footnote{\label{fn:Deltaboundary}We could also have defined $B_{k\ell}$ by $\overline{\Delta_k}\cap B_\ell$ instead of \eqref{Bkelldef}, but that would have caused a bit of trouble because these sets would not have been disjoint. Our choice \eqref{Bkelldef}, on the other hand, has the consequence, which may at first seem like a drawback, that $\cup_k B_{k\ell}\neq B_\ell$ because we have removed the points on the 2d triangles where two tetrahedra meet. However, the set removed, being a subset of a countable union of 2d triangles, has measure 0 on $\Upsilon$, and for any set $A\subseteq \Sigma$ of measure 0, $\exists(A)$ has measure 0 in $\Gamma(\Sigma)$ and, by Definition~\ref{def:hypersurfaceevolution}, also $P_\Sigma(\exists(A))=0$.}

We say that the outcome matrix $ s $ is compatible with $ L $ (denoted $ s:L $) whenever
\begin{equation}
	\forall \ell \in \{1,\ldots,r\}: \begin{cases} \exists k \in \sK: s_{k\ell} = 1 ~~ &\text{if } L_\ell = 1 \\ \forall k \in \sK: s_{k\ell} = 
0 ~~ &\text{if } L_\ell = 0 \end{cases} 
\label{eq:s:L}
\end{equation}
Let $E_k$ be the hyperplane containing $\Delta_k$. 
The configurations in $E_k$ compatible with outcomes $ s_{k\ell} $ or $ s_k :=(s_{k1},\ldots,s_{kr})$ are then given by
\begin{equation}
	M_{k\ell E_k}(s_{k\ell}) := \begin{cases} \exists(B_{k\ell})\subset \Gamma(E_k) ~~ &\text{if } s_{k\ell} = 1 \\ \emptyset(B_{k\ell})\subset \Gamma(E_k) ~~ &\text{if } s_{k\ell} = 0 \end{cases}, \quad M_{kE_k}(s_k) := \bigcap_{\ell=1}^r M_{k\ell E_k}(s_{k\ell}).
\end{equation}
Likewise,
\begin{equation}
	M_{k\ell \Upsilon}(s_{k\ell}) := \begin{cases} \exists(B_{k\ell})\subset \Gamma(\Upsilon) ~~ &\text{if } s_{k\ell} = 1 \\ \emptyset(B_{k\ell})\subset \Gamma(\Upsilon) ~~ &\text{if } s_{k\ell} = 0 \end{cases}, \quad M_{k\Upsilon}(s_k) := \bigcap_{\ell=1}^r M_{k\ell \Upsilon}(s_{k\ell}).
\end{equation}

It follows that, based on the definition \eqref{MsPdef},
\begin{equation}
	M_{\sB}(L) = \bigcup_{s:L} \bigcap_{k \in \sK} M_{k \Upsilon}(s_k)~~\text{up to a set of measure 0,}
\label{eq:MLMsk2}
\end{equation}
meaning that the symmetric difference between the two sets is a set of measure 0 in $\Gamma(\Upsilon)$. This is the case because, as described in Footnote~\ref{fn:Deltaboundary}, the configurations in the symmetric difference have at least one particle in the 2d set $\partial \Delta_k$ for some $k$.

\subsection{Sequential Detection Process}

We now formulate the definition of the sequential detection process and prove agreement with the Born rule. Fix an ordering of $\sK$, i.e., a bijection $\sK \to \NNN$. For ease of notation, we will simply replace $\sK$ by $\NNN$ using this particular ordering. The detection process is:
\begin{itemize}
\item Set $E_0=\{x^0=0\}$ and $\Psi_0=\Psi_{E_0}$.
\item For each $k$ in the specified order, do:
\begin{itemize}
\item Evolve $\Psi_{k-1}$ to $E_k$.
\item Carry out detections of $B_{k\ell}$ for all $\ell=1,\ldots,r$, i.e., quantum measurements of $P_{E_k}(\exists(B_{k\ell}))$, and collapse accordingly, resulting in the (normalized) state vector $\Psi_k\in\Hilbert_{E_k}$.
\item Repeat.
\end{itemize}
\end{itemize}

Note that by Definition~\ref{def:triangular}, each $B_\ell$ intersects only finitely many $\Delta_k$. Thus, from some $K+1$ onwards, all $B_{k\ell}$ are empty, $s_{k\ell}=0$, and no quantum measurement needs to be carried out in $\Delta_k$. Hence, it suffices to consider finitely many repetitions in the above loop, namely those for $k$ up to $K$.

From the flat Born rule and the flat collapse rule, we can now express the detection probabilities and the collapsed state vectors. Fix some $k$ 
and $\ell$; suppose that in the previous tetrahedra $k'<k$ (i.e., none if 
$k=1$), the measurements have already been carried out with outcomes $s_{k'\ell'}$; suppose further that in the previous detector regions $B_{k\ell'}$ with $\ell'<\ell$ (i.e., none if $\ell=1$) in the same tetrahedron $\Delta_k$, the measurements have already been carried out with outcomes $s_{k\ell'}$; suppose further that $\Psi_{k,\ell-1}$ is the collapsed wave function after the previous measurements, which for $\ell>1$ is given by the previous step, for $\ell=1$ and $k>1$ is given by
\be
\Psi_{k,0}= U_{E_{k-1}}^{E_k} \, \Psi_{k-1,r}
\ee 
(with $\Psi_{k-1,r}=\Psi_{k-1}$ in the notation of the process description above), and for $\ell=1,k=1$ is given by
\be
\Psi_{1,0}= U_{E_0}^{E_1} \, \Psi_0\,.
\ee
Conditional on the previous detection outcomes, the probability distribution of the next detection outcome $s_{k\ell}$ is, by the flat Born rule,
\be
\Prob(s_{k\ell}=1) = \bigl\|P_{E_k}(\exists(B_{k\ell}))\, \Psi_{k,\ell-1}\bigr\|^2\,,
\ee
and the state vector collapses, by the flat collapse rule, to
\be
\Psi_{k\ell}= \frac{P_{E_k}(M_{k\ell E_k}(s_{kl}))\,\Psi_{k,\ell-1}}{\|P_{E_k}(M_{k\ell E_k}(s_{kl}))\,\Psi_{k,\ell-1}\|}\,.
\ee
This completes the definition of the sequential detection process.

\bigskip

\begin{lemma}\label{lemma:ink}
Assume the flat Born rule and the flat collapse rule.
Conditional on the measurements in the tetrahedra $k'<k$, the joint distribution of all outcomes $(s_{k\ell})_{\ell=1..r}=s_k$ in $\Delta_k$ is
\be
\Prob( s_{k1},\ldots,s_{kr}) = \bigl\|P_{E_k}(M_{kE_k}(s_k))\, \Psi_{k0}\bigr\|^2,
\ee
and the collapsed wave function after the $kr$-measurement, given $s_k$ with nonzero probability, is
\be
\Psi_{kr}= \frac{P_{E_k}(M_{kE_k}(s_k))\,\Psi_{k0}}{\| P_{E_k}(M_{kE_k}(s_k))\,\Psi_{k0}\|}\,.
\ee
\end{lemma}

\bigskip

\begin{proof}
It is well known general facts about PVMs $P$ that 
\be
P(S_1)\, P(S_2) = P(S_2) \, P(S_1) = P(S_1\cap S_2)
\ee
and that a quantum measurement of $P(S_1)$ with outcome $s_1$ on $\Psi$, followed by one of $P(S_2)$ with outcome $s_2$, have joint Born distribution
\begin{align}
\Prob(s_1=1,s_2=1) &= \Prob(s_2=1|s_1=1)\Prob(s_1=1)\\
&= \Bigl\|P(S_2)\frac{P(S_1)\Psi}{\|P(S_1)\Psi\|}  \Bigr\|^2 \|P(S_1)\Psi\|^2
=\Bigl\| P(S_1\cap S_2) \Psi\Bigr\|^2 \label{joint}
\end{align}
and collapsed state vector, given $s_1=1,s_2=1$,
\be
\Psi' = P(S_2)\frac{P(S_1)\Psi}{\|P(S_1)\Psi\|}/\Bigl\| P(S_2)\frac{P(S_1)\Psi}{\|P(S_1)\Psi\|}\Bigr\| = \frac{P(S_1\cap S_2) \Psi}{\|P(S_1\cap S_2) \Psi\|}\,.
\ee
Iteration with $r$ sets rather than 2 and the definition of $M_{k E_k}(s_k)$ yield Lemma~\ref{lemma:ink}.
\end{proof}

\bigskip

\begin{lemma}\label{lemma:Pksevolution}
(IL) implies that
\begin{equation}
	U_{E_k}^{\Upsilon} P_{E_k} (M_{k E_k}(s_k)) U_{\Upsilon}^{E_k} = P_\Upsilon (M_{k \Upsilon}(s_k)).
\label{eq:Pksevolution}
\end{equation}
\end{lemma}

\bigskip

\begin{proof}
Decompose $ \mathscr{H}_{E_k} = \mathscr{H}_{\Delta_k} \otimes \mathscr{H}_{E_k \setminus \Delta_k} $ and $ \mathscr{H}_{\Upsilon} = \mathscr{H}_{\Delta_k} \otimes \mathscr{H}_{\Upsilon \setminus \Delta_k} $. By (IL), we have that
\begin{equation}\label{UIV2}
	U_{\Upsilon}^{E_k} = I_{\Delta_k} \otimes V_{\Upsilon \setminus \Delta_k}^{E_k \setminus \Delta_k}.
\end{equation}
We know that $\Gamma(E_k)=\Gamma(\Delta_k) \times \Gamma(E_k\setminus \Delta_k)$. The set $M_{k E_k}(s_k)\subseteq \Gamma(E_k)$ factorizes in the same way:
\begin{equation}
	M_{k E_k}(s_k) = N_{k \Delta_k}(s_k) \times \Gamma(E_k \setminus \Delta_k)\,.
\end{equation}
That is because whether a configuration $ q $ is compatible with the outcome $ s_k $, i.e., $ q \in M_{E_k}(s_k) $, does not depend on the points in $ q $ outside of $ \Delta_k $. Here, the set $ N_{k \Delta_k}(s_k) \subseteq \Gamma(\Delta_k) $ is defined in the analogous way to $M_{k E_k}(s_k)$, i.e.,
\begin{equation}\label{Ndef}
	N_{k \Delta_k}(s_k) := \bigcap_{\ell = 1}^r N_{k\ell \Delta_k}(s_{k\ell}), \qquad N_{k\ell \Delta_k}(s_{k\ell}) := \begin{cases} \exists_{\Delta_k}(B_{k\ell}) \quad &\text{if } s_{k\ell} = 1 \\ \emptyset_{\Delta_k}(B_{k\ell})\quad &\text{if } s_{k\ell} = 0, \end{cases}
\end{equation}
where $\exists_A(B)$ means the set of all configurations in $\Gamma(A)$ with at least one particle in $B$. Hence, the projection $ P_{E_k} (M_{k E_k}(s_k)) $ decomposes into a tensor product
\begin{equation}\label{PNk}
	P_{E_k} (M_{k E_k}(s_k)) = P_{\Delta_k}(N_{k \Delta_k}(s_k)) \otimes I_{E_k \setminus \Delta_k}\,,
\end{equation}
and by \eqref{UIV2}, 
\begin{equation}
\begin{aligned}
	U_{E_k}^{\Upsilon} P_{E_k} (M_{k E_k}(s_k)) U_{\Upsilon}^{E_k}
	&= [I_{\Delta_k} \otimes V^{\Upsilon \setminus \Delta_k}_{E_k \setminus \Delta_k}] [P_{\Delta_k}(N_{k \Delta_k}(s_k)) \otimes I_{E_k \setminus \Delta_k}] [I_{\Delta_k} \otimes V_{\Upsilon \setminus \Delta_k}^{E_k \setminus \Delta_k}]\\
	&= [I_{\Delta_k} \circ P_{\Delta_k}(N_{k \Delta_k}(s_k))  \circ I_{\Delta_k}] \; \otimes \;   [V_{E_k \setminus \Delta_k}^{\Upsilon \setminus \Delta_k} \circ I_{E_k \setminus \Delta_k} \circ  V_{\Upsilon \setminus \Delta_k}^{E_k \setminus \Delta_k}]\\
	&= P_{\Delta_k}(N_{k \Delta_k}(s_k)) \; \otimes \; I_{\Upsilon \setminus \Delta_k}\\[1mm]
	&= P_\Upsilon (M_{k \Upsilon}(s_k))
\end{aligned}
\end{equation}
for the same reasons as \eqref{PNk}.
\end{proof}

\bigskip

\begin{prop}\label{prop:sequential3}
Assume the flat Born rule, the flat collapse rule, and (IL).
The unconditional joint distribution of all outcomes, i.e., of the matrix $s$ comprising all $s_{k\ell}$, agrees with the Born distribution on $\Upsilon$,
\be
\Prob(s)= \Bigl\| P_\Upsilon\Bigl(\bigcap_{k\in\NNN} M_{k\Upsilon}(s_k)\Bigr)\, \Psi_\Upsilon\Bigr\|^2
\ee
with $\Psi_\Upsilon=U^\Upsilon_{E_0} \Psi_0$ (actually regardless of whether $\partial B_\ell$ are null sets).
In particular, the distribution of $L=(L_1,\ldots,L_r)$ is the Born distribution $\|P_\Upsilon(M_{\sB\Upsilon}(L))\Psi_\Upsilon\|^2$.
\end{prop}

\bigskip

\begin{proof}
As noted before, all $s_{k\ell}$ vanish from some $K+1$ onwards (and formulas below will take for granted they do), and we need consider only $k\leq K$. The fact, used before in \eqref{joint}, that for subsequent measurements the projections multiply, yields from Lemma~\ref{lemma:ink} that
\be\label{Psklfull}
\Prob(s)
= \Bigl\| 
U_{E_{K}}^\Upsilon
P_{E_{K}}(M_{K E_{K}}(s_{K}))
U_\Upsilon^{E_{K}}\cdots
U_{E_1}^\Upsilon
P_{E_1}(M_{1E_1}(s_1))
U_\Upsilon^{E_1}\Psi_\Upsilon\Bigr\|^2\,.
\ee
Inserting \eqref{eq:Pksevolution} in \eqref{Psklfull} yields
\be\label{Psklfull2}
\begin{aligned}
\Prob(s)
&= \Bigl\| 
P_{\Upsilon}(M_{K \Upsilon}(s_{K}))
\cdots
P_{\Upsilon}(M_{1 \Upsilon}(s_1))
\Psi_\Upsilon\Bigr\|^2\\
&= \Bigl\| P_{\Upsilon}\Bigl(\bigcap_{k=1}^K M_{k \Upsilon}(s_k)\Bigr)
\Psi_\Upsilon\Bigr\|^2\\
&= \Bigl\| P_{\Upsilon}\Bigl(\bigcap_{k\in\NNN} M_{k \Upsilon}(s_k)\Bigr)
\Psi_\Upsilon\Bigr\|^2\,,
\end{aligned}
\ee
as claimed.
\end{proof}

\bigskip

Proposition~\ref{prop:sequential2}, insofar as it concerns the sequential 
detection process, follows from Proposition~\ref{prop:sequential3} (actually regardless of whether $\partial B_\ell$ are null sets), and Proposition~\ref{prop:sequential1} follows further as the special case in which $r=1$, $B_1=B$, and $B_{r+1}=B^c$.

\subsection{Parallel Detection Process}

We now formulate the definition of the parallel detection process and prove the Born rule for it. Throughout the whole subsection, $(IL)$ is assumed.

The proof of Lemma~\ref{lemma:Pksevolution} also shows that, analogously to \eqref{eq:Pksevolution},
\be\label{Pksevolution2}
	U_{E_k}^{\Upsilon} P_{E_k} (M_{k\ell E_k}(s_{k\ell})) U_{\Upsilon}^{E_k} 
= P_\Upsilon (M_{k\ell \Upsilon}(s_{k\ell})).
\ee
As outlined in Section~\ref{sec:introdetectionprocess}, the idea is to think of the detection attempt in $B_{k\ell}$ as a quantum measurement of the observable
\be\label{parallelobservables}
U_{E_k}^{\Upsilon} P_{E_k} (\exists(B_{k\ell})) U_{\Upsilon}^{E_k}=P_\Upsilon (\exists(B_{k\ell}))\,,
\ee
which is \eqref{Pksevolution2} for $s_{k\ell}=1$. Since $B_{k\ell}$ is non-empty only for finitely many $k$ (for $k=1,\ldots,K$), we are considering only finitely many observables. They commute because projections belonging to the same PVM always commute. Their simultaneous measurement is the definition of the parallel detection process.

We now prove the Born rule for the parallel detection process.
When considering the simultaneous measurement of the operators \eqref{parallelobservables}, we need their joint diagonalization; the joint eigenspace with eigenvalues $(s_{k\ell})_{k\ell}$ is the range of
\be
P_\Upsilon\Bigl(\bigcap_{k=1}^K \bigcap_{\ell=1}^r M_{k\ell \Upsilon}(s_{k\ell})\Bigr)
= P_\Upsilon\Bigl(\bigcap_{k=1}^K M_{k \Upsilon}(s_{k})\Bigr)\,,
\ee
so the probability of the outcomes $(s_{k\ell})_{k\ell}$ is
\be
\Bigl\| P_\Upsilon\Bigl(\bigcap_{k=1}^K M_{k \Upsilon}(s_{k})\Bigr) \, \Psi_\Upsilon\Bigr\|^2\,,
\ee
and the probability of outcome $L$ is
\be
\begin{aligned}
\sum_{s:L}\Bigl\| P_\Upsilon\Bigl(\bigcap_{k=1}^K M_{k \Upsilon}(s_{k})\Bigr) \, \Psi_\Upsilon\Bigr\|^2
&= \Bigl\| \sum_{s:L} P_\Upsilon\Bigl(\bigcap_{k=1}^K M_{k \Upsilon}(s_{k})\Bigr) \, \Psi_\Upsilon\Bigr\|^2\\
&= \Bigl\| P_\Upsilon\Bigl(\bigcup_{s:L} \bigcap_{k=1}^K M_{k \Upsilon}(s_{k})\Bigr) \, \Psi_\Upsilon\Bigr\|^2\\
&= \Bigl\| P_\Upsilon\Bigl( M_{\sB}(L)\Bigr) \, \Psi_\Upsilon\Bigr\|^2
\end{aligned}
\ee
because the sets $\bigcap_{k=1}^K M_{k \Upsilon}(s_{k})$ are mutually disjoint and thus their projections are mutually orthogonal, and because of \eqref{eq:MLMsk2} and property (i) in Definition~\ref{def:hypersurfaceevolution}. That is, the probability of outcome $L$ agrees with the Born rule. This proves the statement about the parallel detection process in Proposition~\ref{prop:sequential2} and thus also in Proposition~\ref{prop:sequential1}.

Another way of looking at the parallel detection process is based on tensor products: Since $\Upsilon=\bigcup_{k=1}^K \Delta_k \cup R$ with remainder set $R=\Upsilon\setminus \bigcup_{k=1}^K \Delta_k$, we have from Remark~\ref{rem:Pfactorfinite} in Section~\ref{sec:hypersurfaceevolution} that
\be\label{Upsilontensorproduct}
\Hilbert_\Upsilon = \bigotimes_{k=1}^K \Hilbert_{\Delta_k} \otimes \Hilbert_R \,.
\ee
By (IL), each $\Hilbert_{\Delta_k}$ can be regarded as a factor in $\Hilbert_{E_k} = \Hilbert_{\Delta_k} \otimes \Hilbert_{E_k \setminus \Delta_k}$. With the flat Born rule in mind, or with the idea that $P_{E_k}$ is the configuration observable on $E_k$, the attempted detection in $B_{k\ell}$ can be regarded as a quantum measurement in $\Hilbert_{E_k}$ of the observable $P_{E_k}(\exists(B_{k\ell}))$, which is of the form
\be
P_{E_k}(\exists_{E_k}(B_{k\ell})) = P_{\Delta_k}(\exists_{\Delta_k}(B_{k\ell})) \otimes I_{E_k\setminus \Delta_k}\,,
\ee
Thus, the attempted detection in $B_{k\ell}$ can \emph{also} be regarded as a quantum measurement in $\Hilbert_{\Delta_k}$ of the observable $P_{\Delta_k}(\exists_{\Delta_k}(B_{k\ell}))$. These observables commute for different $\ell$ and equal $k$ because they belong to the same PVM $P_{\Delta_k}$, and they commute for different 
$k$ in $\Hilbert_\Upsilon$ because of the tensor product structure \eqref{Upsilontensorproduct}. It follows that
\be
P_\Upsilon(M_{\sB}(L)) = \sum_{s:L} \bigotimes_{k=1}^K P_{\Delta_k}(N_{k \Delta_k}(s_k)) \otimes I_R
\ee
with $N_{k \Delta_k}$ as in \eqref{Ndef}, which agrees again with the Born rule on $\Upsilon$, as claimed in Proposition~\ref{prop:sequential2}.

\section{Approximation by Triangular Surfaces}
\label{sec:approxtriangular}

In this section, we prove Propositions~\ref{prop:approx} and \ref{prop:uniform}.

\bigskip

\begin{proof}[Proof of Proposition \ref{prop:approx}]
Fix an $n\in\NNN$ and set $\varepsilon=3^{-n}$.
We construct a $ 3 \varepsilon$-approximation $ \Upsilon_n $ to $ \Sigma $. First, consider the function $ f_t: \MMM \to \MMM, \; (x^0,\bx) \mapsto (x^0 - t, \bx) $, which ``lowers a point by an amount $ t $ in time.'' We use $ f $ to define the sets (see Figure \ref{fig:Cauchyconvergence}):
\begin{equation}
	\Sigma_{2\varepsilon} := f_{2 \varepsilon} [\Sigma], \qquad 
	\Sigma_{\varepsilon..3\varepsilon} := \bigcup_{\varepsilon < \varepsilon' < 3\varepsilon} f_{\varepsilon'} [\Sigma].
\end{equation}
So $ \Sigma_{2\varepsilon} $ is a version of $ \Sigma $, lowered by $ 2 \varepsilon $ and $ \Sigma_{\varepsilon..3\varepsilon} $ is a slice below $ \Sigma $ of thickness $ 2 \varepsilon $, centered at $ \Sigma_{2\varepsilon} $.\\

\begin{minipage}{0.54\textwidth}
    \includegraphics[scale=1.0]{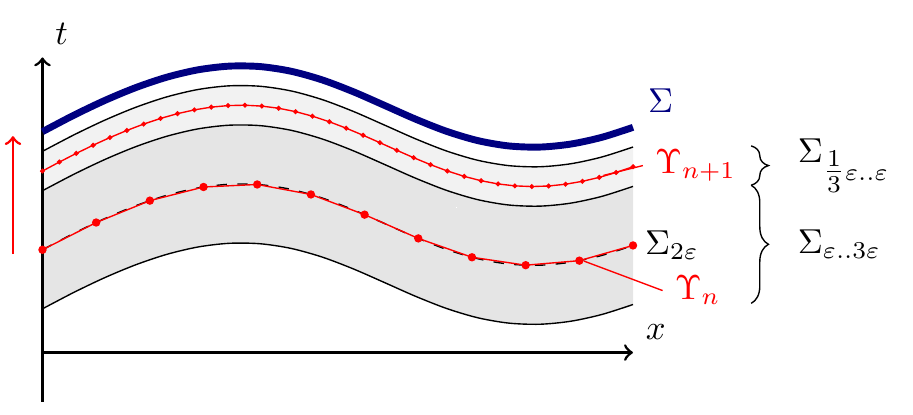}
	\captionof{figure}{Construction of the approximating sequence $ \Upsilon_n \nearrow \Sigma $. Color online.}
    \label{fig:Cauchyconvergence}
\end{minipage}
\hfill
\begin{minipage}{0.36\textwidth}
	\centering
	\scalebox{0.8}{\begin{tikzpicture}

\fill (3.5,3.5) circle (0.1) node[anchor = south west] {$ x^n_{k,i} $};

\filldraw[thick, fill = red, fill opacity = 0.1] (0,2) -- (3,1.5) -- (3.5,3.5) -- cycle;
\filldraw[thick, fill = blue, fill opacity = 0.1] (0,2) .. controls (1,1.3) and (2,1.1) .. (3,1.5) .. controls (3.4,2.5) and (3.3,2.5) .. (3.5,3.5) .. controls (2.5,2.5) and (1.2,1.7) .. cycle;

\draw[thick,->] (0,-0.5) -- (0,4) node[anchor = south west] {$t$};
\draw[thick,->] (-0.5,0.1) -- (3.5,-0.7) node[anchor = south west] {$x^1$};
\draw[thick,->] (-0.3,-0.2) -- (1.5,1) node[anchor = west] {$x^2$};

\draw[red] (1,2.3) -- ++(-0.2,0.5) node[anchor = south] {$ \Delta^n_k \subset \Upsilon_n $};
\draw[blue] (0.8,1.6) -- ++(-0.3,-0.3) node[anchor = north] {$ \Sigma_{2 \varepsilon} $};
\draw (0.8,0) -- ++(-0.3,-0.3) node[anchor = north] {$ \tilde \Delta^n_k $};

\draw[thick, red] (3.5,3.5) -- (2,2.6);
\draw[thick, blue] (3.5,3.5) .. controls (2.8,2.5) and (2.6,2.3) .. (2,2);
\draw[thick] (3.5,0.8) -- (2,0.3);

\filldraw[thick, fill opacity = 0.2] (0,0) -- (3,-0.6) -- (3.5,0.8) -- cycle;
\draw[thick] (2,2.6) -- (2,2);
\draw[dashed] (2,2.6) -- (2,0.3);
\draw[dashed] (3.5,3.5) -- (3.5,0.8);

\fill[red] (2,2.6) circle (0.05) node[anchor = south] {$ y $};
\fill (2,0.3) circle (0.05) node[anchor = north] {$ \boldsymbol{y} $};
\fill (3.5,0.8) circle (0.05) node[anchor = north west] {$ \pi(x^n_{k,i}) $};

\draw[dashed] (2,2.6) -- ++(2.5,-0.5);
\draw[dashed] (2,2) -- ++(2.5,-0.5);
\draw[thick, <->] (4.5,2.1) -- node[anchor = west] {$ h(\boldsymbol{y}) $} (4.5,1.5);

\end{tikzpicture}}
	
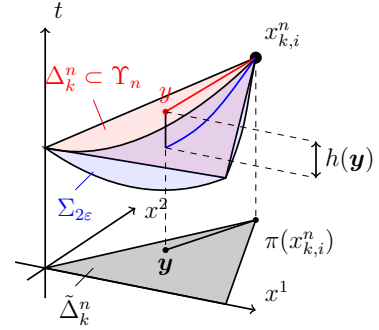
\captionof{figure}{$ |h(\by)| < \varepsilon $ illustrated in 2+1 dim. Color online.}
    \label{fig:Cauchyconvergence2}
\end{minipage}

\vspace{0.5cm}

We now choose a decomposition of $ \RRR^3 $ into (non-regular) tetrahedra 
$ \RRR^3 = \bigcup_{k \in \NNN} \overline{\tilde{\Delta}^n_k} $ with open $ \tilde{\Delta}^n_k $ such that each pair of vertices $ \bx^n_{k,i}, \bx^n_{k,j}, i, j \in \{1,2,3,4\} $ has a distance $ \Vert \bx^n_{k,i} - \bx^n_{k,j} \Vert \leq \varepsilon $ and such that every bounded region intersects only finitely many tetrahedra. For example, we may subdivide $\RRR^3$ into axiparallel cubes with vertices on $\frac{\varepsilon}{\sqrt{3}}\ZZZ^3$ and subdivide each cube into $3!$ tetrahedra with vertices on $\frac{\varepsilon}{\sqrt{3}}\ZZZ^3$.

The four space-time points $ x^n_{k,i} := \pi|_{\Sigma_{2 \varepsilon}}^{-1} \bx^n_{k,i} \in \MMM $ (obtained by lifting $ \bx^n_{k,i} $ up to the $ 2 \varepsilon $-surface, with $i=1,2,3,4$) span a spacelike open tetrahedron $ \Delta^n_k$ in $\MMM $. Now set $ \Upsilon_n := \bigcup_{k \in \NNN} \overline{\Delta^n_k} $.

\bigskip

\noindent\underline{Claim:} $ \Upsilon_n $ is a uniform $ \varepsilon $-approximation of $ \Sigma_{2 \varepsilon} $, i.e., $\Upsilon_n \subset\Sigma_{\varepsilon .. 3\varepsilon}$ (see Figure \ref{fig:Cauchyconvergence}). 

\bigskip

\noindent\underline{Proof:} Regard the surfaces $ \Upsilon_n $ and 
$ \Sigma_{2 \varepsilon} $ as the graphs of functions $\RRR^3\to\RRR$, henceforth denoted simply by $ \Upsilon_n(\cdot)$ and $\Sigma_{2 \varepsilon}(\cdot)$; that is, $(\Upsilon_n(\bx),\bx)\in \Upsilon_n$ for all $\bx\in\RRR^3$ and $x=(\Upsilon_n(\pi(x)),\pi(x))$ for all $x\in\Upsilon_n$. Both functions are Lipschitz-continuous with Lipschitz constant 1. Further, there is always a vertex of $ \tilde{\Delta}^n_k$ (possibly several ones) that maximizes $\Upsilon_n(\cdot)$ on $\overline{\tilde{\Delta}^n_k}$ 
(a ``highest'' vertex), and one (or several) that minimizes $\Upsilon_n(\cdot)$ (a ``lowest'' vertex).
Now consider the ``height difference function'' $ h(\bx) = \Upsilon_n(\bx) - \Sigma_{2 \varepsilon}(\bx) $. (It is Lipschitz continuous with Lipschitz constant 2.) For any vertex $ x^n_{k,i} $, we have that $ h(\pi(x^n_{k,i})) = 0 $. And for any other point $ y \in \Delta^n_k $, we have that $|\pi(x^n_{k,i}) - \pi(y)|_{\RRR^3} < \varepsilon $, so by Lipschitz 
continuity,
\begin{equation}
	\Sigma_{2 \varepsilon}(\pi(x^n_{k,i})) - \Sigma_{2 \varepsilon}(\pi(y)) < \varepsilon\,.
\end{equation}
If $ x^n_{k,i} $ is a highest vertex, then
\begin{equation}
\begin{aligned}
	&\Upsilon_n(\pi(x^n_{k,i})) - \Upsilon_n(\pi(y)) > 0\\
	\Rightarrow \quad &h(\pi(x^n_{k,i})) - h(\pi(y)) > -\varepsilon \quad \Leftrightarrow \quad h(\pi(y)) < \varepsilon
\end{aligned}
\end{equation}	
(see Figure~\ref{fig:Cauchyconvergence2}). The same reasoning with a lowest vertex yields $ h(\pi(y)) > -\varepsilon $, so in total $ |h(\pi(y))| < \varepsilon $, which proves the claim.\hfill$\Box$

\bigskip

\noindent\underline{Claim:} $ \Upsilon_n $ is a Cauchy surface. 

\bigskip

\noindent\underline{Proof:} We need to show that $\Upsilon_n$ is intersected exactly once by every causal inextendible curve $\gamma: (-\infty, \infty) \to \MMM$. We regard $\Upsilon_n$ again as the graph of an equally denoted function $\Upsilon_n: \RRR^3 \to \RRR$. Now, consider the height difference function $h(t) = \gamma^0(t) - \Upsilon_n(\pi(\gamma(t)))$, which tells us ``by how much $\gamma$ is above $\Upsilon_n$.'' Since $\Upsilon_n$ consists of spacelike tetrahedra, $\Upsilon_n$ is Lipschitz-continuous with Lipschitz constant $\le1$. As $\gamma$ is timelike-or-lightlike and w.l.o.g.  directed towards the future, we have that $h$ is strictly increasing, so there can be at most one $t$ with $h(t) = 0$. That is, there is at most one intersection of $\gamma$ with $\Upsilon_n$.

On the other hand, an intermediate value argument yields that there must be at least one intersection point: Otherwise, either $h(t)>0$ for all $t$ or $h(t)<0$ for all $t$; w.l.o.g., assume the former case. Since $\Upsilon_n$ is an $\varepsilon$-approximation to $\Sigma_{2\varepsilon}$, we know that $\gamma^0(t)>\Upsilon_n(\pi(\gamma(t)))>\Sigma_{2\varepsilon}(\pi(\gamma(t)))-\varepsilon=\Sigma_{3\varepsilon}(\pi(\gamma(t)))$, which implies that $\gamma$ does not intersect $\Sigma_{3\varepsilon}$, but that is impossible because $\Sigma_{3\varepsilon}$ is a Cauchy surface. \hfill$\Box$

\bigskip

We can now complete the proof of Proposition \ref{prop:approx}.
Since $ \Upsilon_n $ approximates $ \Sigma_{2 \varepsilon} $ up to $\varepsilon$, it approximates $\Sigma$ up to $ 3 \varepsilon $. Furthermore, $ 
\Upsilon_n \subset \Sigma_{\varepsilon.. 3 \varepsilon} $ and $ \Upsilon_{n+1} \subset \Sigma_{\frac 13 \varepsilon.. \varepsilon} $, and since $\Sigma_{\frac 13 \varepsilon.. \varepsilon}$ lies in the future of $\Sigma_{\varepsilon.. 3 \varepsilon}$ while being disjoint from it, $\Upsilon_{n+1}$ lies in the future of $\Upsilon_n$ (see Figure \ref{fig:Cauchyconvergence}). This completes the proof of Proposition~\ref{prop:approx}.
\end{proof}

\bigskip

Proposition~\ref{prop:uniform} follows from the following statement:

\begin{prop}\label{prop:frameindependence}
Let $\varepsilon>0$, $\Sigma$ be a Cauchy surface, $a_\varepsilon:=(\varepsilon,0,0,0)$ the vertical 4-vector of length $\varepsilon$, and $ g : 
\MMM \to \MMM ,\;  g \in \cP_+^{\uparrow} $ a proper Poincar\'e transformation. Then 
\be
	g[\Sigma+a_\varepsilon] \subset \Bigl\{x+(s,0,0,0):x\in g\Sigma, 0< s <\tilde\varepsilon \Bigr\}
\ee
with
\be
	\tilde{\varepsilon}= (\beta \gamma + \gamma) \varepsilon
\ee
with $\beta\in[0,1)$ the boost velocity of $g$ and $\gamma:=(1-\beta^2)^{-1/2}$ (the ``Lorentz factor''). 
\end{prop}

\begin{minipage}{0.42\textwidth}
	\centering
    \scalebox{0.8}{\begin{tikzpicture}
\draw[thick,->] (0,-0.5) -- (0,4) node[anchor = south west]{$ t $};
\draw[thick,->] (0,0) -- (5,0) node[anchor = south west]{$ x^1 $};

\draw[line width = 2, blue] (0,0.5) .. controls (0.5,0.8) and (0.5,0.8) .. (1,1) .. controls (1.5,1.2) and (2,0.5) .. (3,0.5) .. controls (4,0.5) and (4,1.2) .. (5,1) node[anchor = west] {$ \Sigma $};
\draw[line width = 2, blue] (0,2.5) .. controls (0.5,2.8) and (0.5,2.8) .. (1,3) .. controls (1.5,3.2) and (2,2.5) .. (3,2.5) .. controls (4,2.5) and (4,3.2) .. (5,3) node[anchor = west] {$ \Sigma + a_\varepsilon$};
\fill[line width = 2, blue, opacity = .1] (0,0.5) .. controls (0.5,0.8) and (0.5,0.8) .. (1,1) .. controls (1.5,1.2) and (2,0.5) .. (3,0.5) .. controls (4,0.5) and (4,1.2) .. (5,1) -- (5,3) .. controls (4,3.2) and (4,2.5) .. (3,2.5) .. controls (2,2.5) and (1.5,3.2)  .. (1,3) .. controls (0.5,2.8) and (0.5,2.8) .. (0,2.5);

\draw[line width =1,->] (2.5,0.5) -- node[anchor = east] {$ a_\varepsilon $} ++(0,2);
\end{tikzpicture}}
	\captionof{figure}{$\Sigma$ is being translated in the Proof of Proposition~\ref{prop:frameindependence}. Color online.}
    \label{fig:lorentzinvariance1}
\end{minipage}
\begin{minipage}{0.5\textwidth}
	\centering
	\scalebox{0.8}{\def\betA{0.3}
\def\gammA{0.9539}
\def\boost(#1,#2) { ( {\gammA*#1 + \betA*\gammA*#2-1} , { \betA*\gammA*#1 + \gammA* #2-1} ) }

\begin{tikzpicture}
\draw[thick,->] (0,-0.5) -- (0,4) node[anchor = south west]{$ t $};
\draw[thick,->] (0,0) -- (5,0) node[anchor = south west]{$ x^1 $};

\draw[line width = 2, blue] \boost(0,0.5) .. controls \boost(0.5,0.8) and \boost(0.5,0.8) .. \boost(1,1) .. controls \boost(1.5,1.2) and \boost(2,0.5) .. \boost(3,0.5) .. controls \boost(4,0.5) and \boost(4,1.2) .. \boost(5,1) node[anchor = west] {$ g[\Sigma] $};
\draw[line width = 2, blue] \boost(0,2.5) .. controls \boost(0.5,2.8) and \boost(0.5,2.8) .. \boost(1,3) .. controls \boost(1.5,3.2) and \boost(2,2.5) .. \boost(3,2.5) .. controls \boost(4,2.5) and \boost(4,3.2) .. \boost(5,3) node[anchor = west] {$ g[\Sigma] +\Lambda_0 a_\varepsilon$};
\fill[line width = 2, blue, opacity = .1] \boost(0,0.5) .. controls \boost(0.5,0.8) and \boost(0.5,0.8) .. \boost(1,1) .. controls \boost(1.5,1.2) and \boost(2,0.5) .. \boost(3,0.5) .. controls \boost(4,0.5) and \boost(4,1.2) .. \boost(5,1) -- \boost(5,3) .. controls \boost(4,3.2) and \boost(4,2.5) .. \boost(3,2.5) .. controls \boost(2,2.5) and \boost(1.5,3.2)  .. \boost(1,3) .. controls \boost(0.5,2.8) and \boost(0.5,2.8) .. \boost(0,2.5) ; 

\fill \boost(2.5,0.55) circle (0.08) node[anchor = east] {$ x_a $};
\fill \boost(2.5,2.55) circle (0.08) node[anchor = south west] {$ x_c $};
\fill (1.54,2.22) circle (0.08) node[anchor = south east] {$ x_b $};
\draw[line width = 2, blue, dashed] \boost(2.5,2.55) -- ++ (-0.6,0.6);
\draw[blue] (1.7,2.55) -- ++(0.2,0.8) node[anchor = south west]{max. slope};

\draw[dashed] (1.54,2.69) -- ++ (0,-3.14);
\draw[dashed] \boost(2.5,2.55) -- ++ (0,-2.6);
\draw[<->] (1.52,-0.4) -- node[anchor = north] {$ \beta \gamma \; \varepsilon $} (2.1,-0.4);
\draw[dashed] \boost(2.5,0.55) -- ++ (6,0);
\draw[dashed] \boost(2.5,2.55) -- ++ (4,0);
\draw[dashed] (1.54,2.69)  -- ++ (6,0);
\draw[<->] (5.9,0.25) -- node[anchor = west] {$ \gamma \; \varepsilon $} (5.9,2.15);
\draw[<->] (5.9,2.69) -- node[anchor = west] {$ \beta \gamma \; \varepsilon $} (5.9,2.15);
\draw[<->] (7.4,2.69) -- node[anchor = west] {$ \tilde{\varepsilon} $} (7.4,0.25);

\draw[line width =1,->] \boost(2.5,0.5) -- node[anchor = east] {$ \Lambda_0 a_\varepsilon $} \boost(2.5,2.5)  ;
\end{tikzpicture}}
	\captionof{figure}{The same structure after a boost. Color online.}
    \label{fig:lorentzinvariance2}
\end{minipage}

\begin{proof}[Proof of Proposition \ref{prop:frameindependence}]
A Poincar\'e transformation $g$ consists of a translation and a Lorentz transformation $\Lambda$, which in turn consists of a rotation and a subsequent boost $\Lambda_0$. The rotation leaves $a_\varepsilon$ invariant. Thus, $g[\Sigma + a_\varepsilon]= g\Sigma + \Lambda_0 a_\varepsilon$. Without loss of generality, $\Lambda_0$ is a boost in the $x^1$ direction (see Figures~\ref{fig:lorentzinvariance1} and \ref{fig:lorentzinvariance2}),
\be\label{eq:atrafo}
\Lambda_0 = \begin{pmatrix}\gamma & \beta\gamma&&\\ \beta\gamma&\gamma&&\\&&1&\\&&&1 \end{pmatrix}\,, \text{ so }
\Lambda_0 a_\varepsilon = \begin{pmatrix} \gamma \varepsilon\\ \beta \gamma \varepsilon \\0\\0 \end{pmatrix}\,.
\ee

Consider any point $ x_a = (x^0_a,\bx_a) \in g\Sigma$. Denote by $ x_b = (x^0_b,\bx_b) $ the point on $ g[\Sigma+a_\varepsilon] $ right above $ x_a $, $\bx_b=\bx_a$. We want to show that $x_b^0 \leq x^0_a + \tilde{\varepsilon}$. Set $x_c := x_a + \Lambda_0 a_\varepsilon$
Since $ g[\Sigma+a_\varepsilon] $ is a Cauchy surface, any two points on it (such as $x_b$ and $x_c$) must be spacelike separated, so
\begin{equation}
	|x^0_b - x^0_c| \le |\bx_b - \bx_c| = |\bx_a - \bx_c| = \beta \gamma 
\varepsilon.
\end{equation}
Now the triangle inequality implies the desired bound
\begin{equation}
	|x^0_b - x^0_a| \le |x^0_b - x^0_c| + |x^0_c-x^0_a| \le \beta \gamma \varepsilon + \gamma \varepsilon = \tilde{\varepsilon}.
\end{equation}
\end{proof}

\section{Proof of Theorem~\ref{thm:1}}
\label{sec:proofthm1}

Here is a quick outline of the proof. We want to show that
\be
\Prob_{\sB_n}(L) := \bigl\|P_{\Upsilon_n}(M_{\sB_n}(L)) \Psi_{\Upsilon_n}\bigr\|^2
\ee
converges, as $n\to\infty$, to
\be
\Prob_{\sP}(L) := \bigl\|P_\Sigma(M_{\sP}(L)) \Psi_\Sigma \bigr\|^2\,.
\ee
The proof is done by a squeeze-theorem argument: We will define two distributions $ \widehat{\Prob}_n $ and $ \widecheck{\Prob}_n $ on $\{0,1\}^r$ 
such that
\begin{equation}
	\widehat{\Prob}_n(L) \le \Prob_{\sB_n}(L) \le \widecheck{\Prob}_n(L), \qquad \widehat{\Prob}_n(L) \le \Prob_{\sP}(L) \le \widecheck{\Prob}_n(L),
\end{equation}
and prove that $ \widehat{\Prob}_n(L), \widecheck{\Prob}_n(L) $ both converge to $\Prob_{\sP}(L)$ as $n\to\infty$. 

\bigskip

We go through some preparations for the proof. To begin with, it is easy to see that $\sB_n=(B_{n1},\ldots,B_{nr})$ with
\be
B_{n\ell}=\pi_\Sigma^{\Upsilon_n}(P_\ell)
\ee
is an admissible partition of $\Upsilon_n$: First, $B_{n\ell}\cap B_{nm}=\emptyset$ for $\ell\neq m$ because $\pi_\Sigma^{\Upsilon_n}$ is a bijection. Second, $B_{n\ell}$ is bounded because $\pi_\Sigma^{\Upsilon_n}$ maps bounded sets to bounded sets. Third, the boundary $\partial B_{n\ell}$ of $B_{n\ell}$ in $\Upsilon_n$ is $\pi_\Sigma^{\Upsilon_n}(\partial P_\ell)$ because $\pi_\Sigma^{\Upsilon_n}$ is a homeomorphism. Finally, in order to obtain that $\mu_{\Upsilon_n}(\partial B_{n\ell})=0$ we note that 
$\mu_\Sigma(\partial P_\ell)=0$, that $\Sigma$ (and $\Upsilon_n$) possesses a spacelike tangent plane almost everywhere (relative to Lebesgue measure $\lambda$ on $\RRR^3$), and that, at points with a spacelike tangent plane, $\mu_\Sigma$ possesses a nonzero density relative to $\lambda\circ\pi_\Sigma$, so $\mu_\Sigma$ and $\lambda\circ \pi_\Sigma$ have the same null sets.

For the definition of  $ \widehat{\Prob}_n, \widecheck{\Prob}_n$ we introduce more notation: 

We define
\begin{equation}
	\widehat{C}_{n\ell} := \Sr(B_{n\ell},\Sigma), ~~~~~
	\widecheck{C}_{n\ell} := \Gr(B_{n\ell},\Sigma).
\label{eq:Ccheck}
\end{equation}

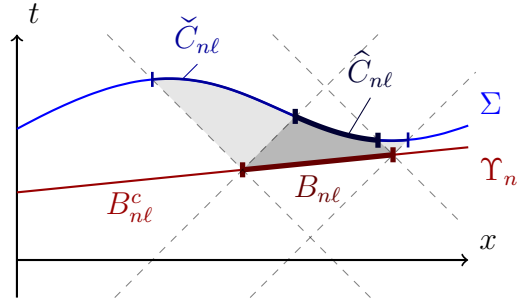
\begin{figure}[hbt]
    \centering
    \scalebox{1.0}{\begin{tikzpicture}[declare function = {
	f(\x) = 0.05*((\x-3.5)^3 - 9*(\x-3.5))*(1/(1+0.1*(\x-3.5)^2)) + 2;}]

\def\Na{30}
\def\xmax{6}
\def\step{\xmax / \Na}

\fill[opacity = 0.1] (3,1.2) -- plot[domain = 1.8:5.2]({\x} , {f(\x)}) -- (5,1.4) -- cycle ;
\fill[opacity = 0.2] (3,1.2) -- plot[domain = 3.7:4.8]({\x} , {f(\x)}) -- (5,1.4) -- cycle ;

\draw[thick,->] (0,-0.5) -- (0,3) node[anchor = south west] {$ t $};
\draw[thick,->] (0,0) -- (6,0) node[anchor = south west] {$ x $};

\draw[thick, red!60!black] (0,0.9) -- (6,1.5) node[anchor = north west] {$ \Upsilon_n $};
\draw[line width = 2, red!40!black] (3,1.2) --node[anchor = north] {$ B_{n\ell} $} (5,1.4) ;
\draw[line width = 2, red!40!black] (3,1.1) -- ++ (0,0.2);
\draw[line width = 2, red!40!black] (5,1.3) -- ++ (0,0.2);

\draw[dashed, opacity = 0.5] (1.3,-0.5) -- ++(3.5,3.5);
\draw[dashed, opacity = 0.5] (3.1,-0.5) -- ++(3,3);
\draw[dashed, opacity = 0.5] (1.2,3) -- ++(3.5,-3.5);
\draw[dashed, opacity = 0.5] (3.4,3) -- ++(2.5,-2.5);

\draw[domain=0:6, samples = 20, smooth, thick, blue] plot({\x}, {f(\x)}) node[anchor = south west] {$ \Sigma $};
\draw[domain=1.8:5.2, samples = 20, smooth, line width = 1, blue!60!black] plot({\x}, {f(\x)});
\draw[line width = 1, blue!60!black] (1.8,{f(1.8)-0.1} ) -- ++ (0,0.2);
\draw[line width = 1, blue!60!black] (5.2,{f(5.2)-0.1} ) -- ++ (0,0.2);
\draw[blue!60!black] (2.2, {f(2.2)} ) -- ++(0.2,0.2) node[anchor = south] {$ \widecheck{C}_{n\ell} $};

\draw[domain=3.7:4.8, samples = 20, smooth, line width = 2, blue!20!black] plot({\x}, {f(\x)});
\draw[line width = 2, blue!20!black] (3.7,{f(3.7)-0.1} ) -- ++ (0,0.2);
\draw[line width = 2, blue!20!black] (4.8,{f(4.8)-0.1} ) -- ++ (0,0.2);
\draw[blue!20!black] (4.4, {f(4.4)} ) -- ++(0.3,0.5) node[anchor = south] {$ \widehat{C}_{n\ell} $};

%\node[blue!60!black] at (3,2.5) {$\perp$};
%\node[blue!60!black] at (5,1.8) {$\perp$};
%\node[blue!20!black] at (4.2,2) {$\parallel$};

\node[red!60!black] at (1.5,0.7) {$B_{n\ell}^c$};

\end{tikzpicture}}
    \caption{Definition of $\widehat{C}_{n\ell}$ and $\widecheck{C}_{n\ell}$. Color online.}
    \label{fig:Supportgrowth}
\end{figure}

The corresponding sets of compatibility in configuration space $\Gamma(\Sigma)$ are
\begin{equation}
	\widehat{M}_{n\ell}(L_{\ell}) := \begin{cases} \exists(\widehat{C}_{n\ell}) ~ &\text{if } L_{\ell} =1 \\ \emptyset(\widecheck{C}_{n\ell})~ &\text{if } L_{\ell} = 0, \end{cases} \qquad
	\widecheck{M}_{n\ell}(L_{\ell}) := \begin{cases} \exists(\widecheck{C}_{n\ell}) ~ &\text{if } L_{\ell} =1 \\ \emptyset(\widehat{C}_{n\ell}) 
~ &\text{if } L_{\ell} = 0, \end{cases}
\label{eq:MchecksigmaP}
\end{equation}
\begin{equation}
	\widehat{M}_{n\Sigma}(L) 
	:= \bigcap_{\ell=1}^r \widehat{M}_{n\ell}(L_{\ell}), \qquad
	\widecheck{M}_{n\Sigma}(L) 
	:= \bigcap_{\ell=1}^r \widecheck{M}_{n\ell}(L_{\ell}).
\label{eq:MchecksigmaP2}
\end{equation}
The probability distributions that serve for the squeeze-theorem bounds are defined by
\begin{equation}
	\widehat{\Prob}_n(L) := \langle \Psi_\Sigma | P_\Sigma(\widehat{M}_{n\Sigma}(L)) | \Psi_\Sigma \rangle \qquad \widecheck{\Prob}_n(L) := \langle \Psi_\Sigma | P_\Sigma(\widecheck{M}_{n\Sigma}(L)) | \Psi_\Sigma \rangle.
\end{equation}

\bigskip

\begin{lemma}[Squeeze-theorem bound for $ \Prob_{\sP} $]~
For all $L\in\{0,1\}^r$,
\begin{align}
	\widehat{M}_{n\Sigma}(L) \subseteq M_\sP&(L) 
	\subseteq \widecheck{M}_{n\Sigma}(L),
	\label{eq:setestimate}\\
	\text{hence}~~P_\Sigma(\widehat{M}_{n\Sigma}(L)) 
	\leq P_\Sigma(M_{\sP}&(L)) 
	\leq P_\Sigma(\widecheck{M}_{n\Sigma}(L)),
	\label{eq:setestimate2}\\
	\text{and}~~~~~~~~~~\widehat{\Prob}_n(L)
	\le \Prob_\sP&(L)
	\le \widecheck{\Prob}_n(L).
	\label{eq:bounds}
\end{align}
\label{lem:setestimate}
\end{lemma}

\begin{proof}
The statement is actually true for any triangular surface $\Upsilon$, regardless of whether it belongs to a sequence converging to $\Sigma$. Since 
we need it for $\Upsilon_n$, we use here the notation that refers to $\Upsilon_n$. 

The inclusion
\be
\widehat{C}_{n\ell} \subseteq P_{\ell} \subseteq \widecheck{C}_{n\ell}
\ee
is obvious, since $ \pi[\widehat{C}_{n\ell}] $ is a shrunk version of $ \pi[P_\ell]$ (i.e., smaller) and $ \pi[\widecheck{C}_{n\ell}] $ is a grown 
version of it (i.e., larger).

We ``lift'' those sets to configuration space, keeping in mind that 
\be\label{inclusionexists}
\text{if }A \subseteq B, \text{ then } \exists(A) \subseteq \exists(B) \text{ and  } \emptyset(A) \supseteq \emptyset(B).
\ee
By definition \eqref{eq:MchecksigmaP} we then have:
\begin{equation}
	\widehat{M}_{n\ell}(L_{\ell}) \subseteq M_{\ell\Sigma}(L_{\ell}) \subseteq \widecheck{M}_{n\ell}(L_{\ell}).
\label{eq:setestimatesmall}
\end{equation}
Inclusions persist under intersections, i.e., 
\be\label{inclusionintersection}
\text{if }A_\ell \subseteq B_\ell \text{ for all $\ell$, then } 
\bigcap_{\ell} A_\ell \subseteq \bigcap_{\ell} B_\ell\,.
\ee
This yields \eqref{eq:setestimate}. The transition from sets $ M $ to projections $ P(M) $ as in \eqref{eq:setestimate2} is  straightforward, and sandwiching between $\Psi_\Sigma$'s yields \eqref{eq:bounds}.
\end{proof}

\bigskip

\begin{lemma}[Squeeze-theorem bound for $\Prob_{\sB_n}$]~\label{lem:setestimate2}
Assume (PL). Then, for all $L\in\{0,1\}^r$,
\begin{align}
	P_\Sigma(\widehat{M}_{n\Sigma} (L))
	\le  U_{\Upsilon_n}^{\Sigma} P_{\Upsilon_n}(M_{\sB_n} &(L)) \,
	U_{\Sigma}^{\Upsilon_n}
	\le P_\Sigma(\widecheck{M}_{n\Sigma} (L)),\label{Pestimate}\\
	\text{hence}~~~~~~~~\widehat{\Prob}_n(L) 
	\le \Prob_{\sB_n}&(L)
	\le \widecheck{\Prob}_n(L).
	\label{eq:bounds2}
\end{align}
\end{lemma}

\begin{proof}
Also this statement is actually true for any triangular surface $\Upsilon$, regardless of whether it belongs to a sequence converging to $\Sigma$.

By (PL) \eqref{PL},
\be
U^{\Sigma'}_\Sigma \, P_\Sigma (\forall A) \, U^\Sigma_{\Sigma'} 
\leq P_{\Sigma'}(\forall \Gr(A,\Sigma'))\,.
\ee
Since $(\exists A)^c = \emptyset A = \forall(A^c)$, we have that
\be
\begin{aligned}
U^{\Sigma'}_\Sigma \, P_\Sigma (\exists A) \, U^\Sigma_{\Sigma'} 
&=U^{\Sigma'}_\Sigma \, (I-P_\Sigma ((\exists A)^c)) \, U^\Sigma_{\Sigma'} \\
&=U^{\Sigma'}_\Sigma \, (I-P_\Sigma (\forall(A^c))) \, U^\Sigma_{\Sigma'} \\
&\geq I-P_{\Sigma'}(\forall \Gr(A^c,\Sigma'))\\
&= I-P_{\Sigma'}(\forall (\Sr(A,\Sigma')^c))\\
&= I-P_{\Sigma'}((\exists \Sr(A,\Sigma'))^c)\\
&= P_{\Sigma'}(\exists \Sr(A,\Sigma'))
\end{aligned}
\ee
and
\be
\begin{aligned}
U^{\Sigma'}_\Sigma \, P_\Sigma (\emptyset A) \, U^\Sigma_{\Sigma'} 
&=U^{\Sigma'}_\Sigma \, P_\Sigma (\forall(A^c)) \, U^\Sigma_{\Sigma'} \\
&\leq P_{\Sigma'}(\forall \Gr(A^c,\Sigma'))\\
&= P_{\Sigma'}(\forall (\Sr(A,\Sigma')^c))\\
&= P_{\Sigma'}(\emptyset \Sr(A,\Sigma'))\,.
\end{aligned}
\ee
Thus, inserting $A\to B_{n\ell}$, $\Sigma\to \Upsilon_n$, and $\Sigma'\to\Sigma$,
\be\label{ineq1}
\begin{aligned}
U^{\Sigma}_{\Upsilon_n} \, P_{\Upsilon_n} (\exists B_{n\ell}) \, U^{\Upsilon_n}_{\Sigma} 
&\geq P_{\Sigma}(\exists \widehat{C}_{n\ell})\\
U^{\Sigma}_{\Upsilon_n} \, P_{\Upsilon_n} (\emptyset B_{n\ell}) \, U^{\Upsilon_n}_{\Sigma}
&\leq P_{\Sigma}(\emptyset \widehat{C}_{n\ell})\,.
\end{aligned}
\ee
On the other hand, inserting $A\to \widecheck{C}_{n\ell}$, $\Sigma'\to\Upsilon_n$, and $\Sigma\to \Sigma$,
\be\label{ineq2}
\begin{aligned}
U_{\Sigma}^{\Upsilon_n} \, P_{\Sigma} (\exists \widecheck{C}_{n\ell}) \, U_{\Upsilon_n}^{\Sigma} 
&\geq P_{\Upsilon_n}(\exists \Sr(\widecheck{C}_{n\ell},\Upsilon_n))\\
U_{\Sigma}^{\Upsilon_n} \, P_{\Sigma} (\emptyset \widecheck{C}_{n\ell}) \, U_{\Upsilon_n}^{\Sigma}
&\leq P_{\Upsilon_n}(\emptyset \Sr(\widecheck{C}_{n\ell},\Upsilon_n))\,.
\end{aligned}
\ee
Since for $A\subseteq \Sigma$ always
\be
A\subseteq \Sr(\Gr(A,\Sigma'),\Sigma)\,,
\ee
and since $A\subseteq B$ implies $\exists(A) \subseteq \exists(B)$ and $\emptyset(A) \supseteq \emptyset(B)$, we have that
\be\label{ineq3}
\begin{aligned}
P_{\Upsilon_n}(\exists \Sr(\widecheck{C}_{n\ell},\Upsilon_n)) 
&\geq P_{\Upsilon_n}(\exists B_{n\ell})\\
P_{\Upsilon_n}(\emptyset \Sr(\widecheck{C}_{n\ell},\Upsilon_n))
&\leq P_{\Upsilon_n}(\emptyset B_{n\ell})\,.
\end{aligned}
\ee 
Putting together \eqref{ineq1}, \eqref{ineq2}, \eqref{ineq3},
\be
\begin{aligned}
P_{\Sigma} (\emptyset\widecheck{C}_{n\ell}) 
&\leq U_{\Upsilon_n}^{\Sigma} \, P_{\Upsilon_n}(\emptyset B_{n\ell}) \, U_{\Sigma}^{\Upsilon_n}\leq P_{\Sigma}(\emptyset \widehat{C}_{n\ell})\\
P_{\Sigma}(\exists \widehat{C}_{n\ell})
&\leq U^{\Sigma}_{\Upsilon_n} \, P_{\Upsilon_n} (\exists B_{n\ell}) \, U^{\Upsilon_n}_{\Sigma} 
\leq P_{\Sigma} (\exists\widecheck{C}_{n\ell}),
\end{aligned}
\ee 
that is, in another notation,
\be\label{Pellestimate}
P_\Sigma(\widehat{M}_{n\ell}(L_{\ell}) )
\leq U_{\Upsilon_n}^{\Sigma} \, P_{\Upsilon_n}(M_{\ell\Upsilon_n}(L_{\ell})) \, U_{\Sigma}^{\Upsilon_n}
\leq P_\Sigma(\widecheck{M}_{n\ell}(L_{\ell}))\,.
\ee

Now we want to conclude an analogous statement about $L$ instead of $L_\ell$. Note that $U_{\Upsilon_n}^{\Sigma} \, P_{\Upsilon_n}(\cdot) \, U_{\Sigma}^{\Upsilon_n}$ and $P_\Sigma(\cdot)$ are two different PVMs that will in general not even commute with each other. The argument that we need has the following general form:
For two different PVMs $ P_1, P_2 $, the ranges satisfy the relations
\begin{equation}
\begin{aligned}
	&P_1(A_1) \le P_2(A_2) \; \land \; P_1(B_1) \le P_2(B_2)\\
	\Leftrightarrow \quad &\range(P_1(A_1)) \subseteq \range(P_2(A_2)) \; \land \; \range(P_1(B_1)) \subseteq \range(P_2(B_2))\\
	\Rightarrow \quad &\underbrace{\range(P_1(A_1)) \cap \range(P_1(B_1))}_{= \range(P_1(A_1) P_1(B_1))} \subseteq \underbrace{\range(P_2(A_2)) \cap \range(P_2(B_2))}_{= \range(P_2(A_2) P_2(B_2))}\\
	\Leftrightarrow \quad &P_1(A_1) P_1(B_1) \le P_2(A_2) P_2(B_2)\\
	\Leftrightarrow \quad &P_1(A_1 \cap B_1) \le P_2(A_2 \cap B_2).
\end{aligned}
\end{equation}
Applying this argument to \eqref{Pellestimate} and the finite intersection $ \bigcap_{\ell} $ yields \eqref{Pestimate}.
\end{proof}

\bigskip

\begin{lemma}\label{lem:closure}
Fix $\ell\in\{1,\ldots,r\}$; $\widecheck{C}_{n\ell}$ is a decreasing sequence of sets, $\widecheck{C}_{n\ell} \supseteq \widecheck{C}_{n+1,\ell}$, 
with
\be\label{Pclosure}
\bigcap_{n \in \NNN} \widecheck{C}_{n\ell} \subseteq \overline{P_\ell}.
\ee
$\widehat{C}_{n\ell}$ is an increasing sequence of sets, $\widehat{C}_{n\ell} \subseteq \widehat{C}_{n+1,\ell}$, with 
\be\label{Pinterior}
\bigcup_{n \in \NNN} \widehat{C}_{n\ell} \supseteq \mathrm{interior}_\Sigma(P_\ell).
\ee
In particular, 
\be\label{Pboundary}
\bigcap_{n \in \NNN} \widecheck{C}_{n\ell}\setminus \widehat{C}_{n\ell} \subseteq \partial P_\ell \,.
\ee
Moreover, equality holds in \eqref{Pclosure}, \eqref{Pinterior}, and \eqref{Pboundary} whenever $\Upsilon_n\cap \Sigma=\emptyset$.
\end{lemma}

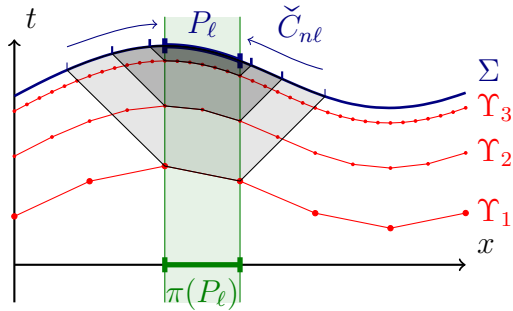
\begin{figure}[hbt]
    \centering
    \scalebox{1.0}{\usetikzlibrary{decorations.pathreplacing}
\begin{tikzpicture}[declare function = {
	f(\x) = 0.05*((\x-3.5)^3 - 9*(\x-3.5))*(1/(1+0.1*(\x-3.5)^2)) + 2.5;}]

\def\dyA{0.2};
\def\dyB{0.8};
\def\dyC{1.6};

\def\Na{30}
\def\xmax{6}
\def\step{\xmax / \Na}

\fill[opacity = 0.1, green!50!black] (2,-0.5) rectangle ++(1,3.8);
\draw[green!50!black] (2,-0.5) -- ++(0,3.8);
\draw[green!50!black] (3,-0.5) -- ++(0,3.8);

\draw[thick,->] (0,-0.5) -- (0,3) node[anchor = south west] {$ t $};
\draw[thick,->] (0,0) -- (6,0) node[anchor = south west] {$ x $};

\draw[ line width = 2, green!50!black] (2,0) -- node[anchor = north] {$ \pi(P_\ell) $} (3,0);
\draw[ line width = 2, green!50!black] (2,-0.1) -- ++(0,0.2);
\draw[ line width = 2, green!50!black] (3,-0.1) -- ++(0,0.2);

\draw[domain=0:6, samples = 20, smooth, line width = 1, blue!50!black] plot({\x}, {f(\x)}) node[anchor = south west] {$ \Sigma $};
\draw[domain=2:3, samples = 20, smooth, line width = 2, blue!50!black] plot({\x}, {f(\x)}) ;
\draw[ line width = 2, blue!50!black] (2,{f(2)-0.1}) -- ++(0,0.2);
\draw[ line width = 2, blue!50!black] (3,{f(3)-0.1}) -- ++(0,0.2);
\node[blue!50!black!] at (2.5,3.2) {$ P_\ell $};

\draw[domain = 0:6, samples = 7, red, mark = *, mark size = 1pt] plot({\x}, {f(\x)-\dyC})node[anchor = west] {$ \Upsilon_1 $} ;
\draw[domain = 0:6, samples = 13, red, mark = *, mark size = 0.5pt] plot({\x}, {f(\x)-\dyB})node[anchor = west] {$ \Upsilon_2 $} ;
\draw[domain = 0:6, samples = 37, red, mark = *, mark size = 0.5pt] plot({\x}, {f(\x)-\dyA})node[anchor = west] {$ \Upsilon_3 $} ;

\filldraw[domain = 0.7:4.13, fill opacity = 0.1] plot[domain = 0.7:4.13]({\x}, {f(\x)}) -- plot[domain = 3:2, samples = 2]({\x}, {f(\x)-\dyC})  -- cycle;
\filldraw[fill opacity = 0.2] plot[domain = 1.3:3.56]({\x}, {f(\x)}) -- plot[domain = 3:2, samples = 3]({\x}, {f(\x)-\dyB}) -- cycle;
\filldraw[fill opacity = 0.3] plot[domain = 1.8:3.15]({\x}, {f(\x)}) -- plot[domain = 3:2, samples = 7]({\x}, {f(\x)-\dyA}) -- cycle;
%\draw[<-, black!50!white] (2.2,1.4) -- ++(0.2,-0.3) node[anchor = north west] {$ A_1 $};
%\draw[<-, black!70!white] (2.2,2.2) -- ++(0.2,-0.3) node[anchor = north west] {$ A_2 $};
%\draw[<-] (2.2,2.8) -- ++(0.2,-0.3) node[anchor = north west] {$ A_3 $};

\draw[blue!50!black] (0.7,{f(0.7)}) -- ++(0,0.1);
\draw[blue!50!black] (4.13,{f(4.13)}) -- ++(0,0.1);
\draw[thick, blue!50!black] (1.3,{f(1.3)}) -- ++(0,0.1);
\draw[thick, blue!50!black] (3.56,{f(3.56)}) -- ++(0,0.1);
\draw[line width = 1, blue!50!black] (1.8,{f(1.8)}) -- ++(0,0.1);
\draw[line width = 1, blue!50!black] (3.15,{f(3.15)}) -- ++(0,0.1);

\draw[->, blue!50!black] (0.7,2.9) .. controls (1,3) and (1.5,3.2) .. (1.9,3.2);
\draw[->, blue!50!black] (4.1,2.6) .. controls (3.7,2.7) and (3.4,2.9) .. (3.1,3);
\node[blue!50!black!] at (3.8,3.2) {$ \widecheck{C}_{n\ell} $};

\end{tikzpicture}}
    \caption{Convergence of the sets $\widecheck{C}_{n\ell}$ as $n\to\infty$ for fixed $\ell$ as in Lemma~\ref{lem:closure}. Color online.}
    \label{fig:Anset}
\end{figure}

\begin{proof}
The decreasing/increasing behavior of the sequence is a direct consequence of $\Upsilon_{n+1}\subseteq \future(\Upsilon_n)$ and the definition of grown and shrunk set. For demonstrating \eqref{Pclosure}, since $\pi_\Sigma$ is a homeomorphism $\Sigma\to\RRR^3$, it suffices to show that $\bigcap_n \pi(\widecheck{C}_{n\ell}) \subseteq \overline{\pi(P_\ell)}$ in $\RRR^3$. If $\by\notin \overline{\pi(P_\ell)}$, then it has positive distance to $\pi(P_\ell)$ and $\pi[[\Sigma-(\varepsilon,0,0,0)] \cap \past(\pi_\Sigma^{-1}(\by))]$ is disjoint from $\pi(P_\ell)$ for sufficiently small $\varepsilon>0$, so $\by\notin \pi(\widecheck{C}_{n\ell})$ for sufficiently large $n$. Similar arguments yield \eqref{Pinterior}.
Concerning the statement about equality, in that case for every $x\in B_{n\ell}$, $\future(x) \cap \past(\Sigma)$ has nonempty interior in $\MMM$, 
so $\pi(\widecheck{C}_{n\ell})$ contains an open neighborhood of $\pi(P_\ell)$ and thus $\overline{\pi(P_\ell)}$. Similarly for the interior.
\end{proof}

\bigskip

\begin{lemma}\label{lem:nullset}
For every $L\in\{0,1\}^r$, ~ $\displaystyle\bigcap_{n \in \NNN} {\widecheck{M}_{n\Sigma}(L)} \setminus \widehat{M}_{n\Sigma}(L) $ ~ is a null set w.r.t.\ $\mu_{\Gamma(\Sigma)}$.
\end{lemma}

\begin{proof}
We make use here of the requirement $\mu_\Sigma(\partial P_\ell)=0$ in Definition~\ref{def:admissible}.
Consider first $ {\widecheck{M}_{n\ell}(L_\ell)} $ and $ {\widehat{M}_{n\ell}(L_\ell)} $. In case $ L_\ell = 1 $, we have that
\begin{equation}\label{Mnl1}
\begin{aligned}
	&{\widecheck{M}_{n\ell}(1)} = \exists \widecheck{C}_{n\ell}, \qquad
	\widehat{M}_{n\ell}(1) = \exists \widehat{C}_{n\ell} \\
	\Rightarrow \quad &{\widecheck{M}_{n\ell}(1)} \setminus \widehat{M}_{n\ell}(1) =(\exists \widecheck{C}_{n\ell}) \cap (\emptyset \widehat{C}_{n\ell}).
\end{aligned}
\end{equation}

In case $ L_\ell = 0 $, we have that
\begin{equation}
	{\widecheck{M}_{n\ell}(0)} \setminus \widehat{M}_{n\ell}(0) = (\emptyset\widehat{C}_{n\ell}) \cap (\exists\widecheck{C}_{n\ell}).
\end{equation}
So either way,
\be
	{\widecheck{M}_{n\ell}(L_\ell)} \setminus \widehat{M}_{n\ell}(L_\ell) = 
(\emptyset\widehat{C}_{n\ell}) \cap (\exists\widecheck{C}_{n\ell}) \subseteq  \exists(\widecheck{C}_{n\ell} \setminus {\widehat{C}_{n\ell}}).
\ee
Now we want to consider $L$ instead of $L_\ell$. It is a general fact about sets that if $A_\ell \subseteq B_\ell$ for all $\ell$, then
\be
\Bigl(\bigcap_\ell B_\ell\Bigr)\setminus \Bigl(\bigcap_\ell A_\ell\Bigr) \subseteq \bigcup_\ell (B_\ell \setminus A_\ell). 
\ee
Thus, for $A_\ell = \widehat{M}_{n\ell}(L_\ell)$ and $B_\ell =\widecheck{M}_{n\ell}(L_\ell)$,
\be
{\widecheck{M}_{n\Sigma}(L)} \setminus \widehat{M}_{n\Sigma}(L)
\subseteq \bigcup_{\ell=1}^r {\widecheck{M}_{n\ell}(L_\ell)} \setminus \widehat{M}_{n\ell}(L_\ell)
\subseteq  \bigcup_{\ell=1}^r\exists(\widecheck{C}_{n\ell} \setminus {\widehat{C}_{n\ell}})
=\exists\Bigl(\bigcup_{\ell=1}^r (\widecheck{C}_{n\ell} \setminus \widehat{C}_{n\ell})\Bigr).
\ee
Now we want to take the intersection over all $n\in\NNN$. In this regard, 
we first note the following extension of \eqref{inclusionexists}: {\it if 
$(A_n)_{n\in\NNN}$ is a decreasing sequence of sets, then} 
\be\label{intersectionexists}
\bigcap_n \exists A_n = \exists\Bigl(\bigcap_n A_n \Bigr).
\ee
After all, if $q$ is a finite set that intersects every $A_n$, then it must contain a point from $\bigcap_n A_n$; conversely, a finite set $q$ intersecting $\bigcap_n A_n$ trivially intersects every $A_n$.

Applying this to $A_n=\bigcup_{\ell} (\widecheck{C}_{n\ell} \setminus \widehat{C}_{n\ell})$, which is decreasing because $\widecheck{C}_{n\ell} \setminus \widehat{C}_{n\ell}$ is, we obtain that
\begin{equation}
	\bigcap_{n \in \NNN} {\widecheck{M}_{n\Sigma}(L)} \setminus \widehat{M}_{n\Sigma}(L) \subseteq  
	\exists \Bigl(\bigcap_{n \in \NNN} \bigcup_{\ell=1}^r \widecheck{C}_{n\ell} \setminus \widehat{C}_{n\ell}\Bigr).
\end{equation}
It is another general fact about sets (not unrelated to \eqref{intersectionexists}) that {\it if for every $\ell\in\{1,\ldots,r\}$, $(A_{n\ell})_{n\in\NNN}$ is a decreasing sequence of sets, then}
\be
\bigcap_{n\in\NNN} \bigcup_{\ell=1}^r A_{n\ell} = \bigcup_{\ell=1}^r \bigcap_{n\in\NNN} A_{n\ell}\,.
\ee
Thus, for $A_{n\ell}= \widecheck{C}_{n\ell} \setminus \widehat{C}_{n\ell}$,
\begin{equation}
	\bigcap_{n \in \NNN} {\widecheck{M}_{n\Sigma}(L)} \setminus \widehat{M}_{n\Sigma}(L) \subseteq  
	\exists \Bigl(\bigcup_{\ell=1}^r \bigcap_{n \in \NNN}  \widecheck{C}_{n\ell} \setminus \widehat{C}_{n\ell}\Bigr)
	\subseteq \exists \Bigl(\bigcup_{\ell=1}^r \partial P_\ell\Bigr)
\end{equation}
by Lemma~\ref{lem:closure} and \eqref{inclusionexists}. For any set $A$ with $\mu_\Sigma(A)=0$ it follows that $\exists A$ is, in every sector of configuration space $\Gamma(\Sigma)$, a finite union of null sets, so $\mu_{\Gamma(\Sigma)}(\exists A)=0$. For $A=\bigcup_{\ell} \partial P_\ell$ we obtain the statement of Lemma~\ref{lem:nullset}.
\end{proof}

\bigskip\bigskip

\begin{proof}[Proof of Theorem \ref{thm:1}] ~
By Lemma~\ref{lem:setestimate} and Lemma~\ref{lem:setestimate2}, it suffices to show that for every $L\in\{0,1\}^r$,
\be
\widecheck{\Prob}_n(L) - \widehat{\Prob}_n(L) \to 0 ~~~~ \text{as }n\to\infty. 
\ee

From Lemma~\ref{lem:nullset} and the requirement (i) of Definition~\ref{def:hypersurfaceevolution}, according to which $P_\Sigma$ must be absolutely continuous with respect to $\mu_{\Gamma(\Sigma)}$, we have that
\be
P_\Sigma \left(\bigcap_{n \in \NNN} {\widecheck{M}_{n\Sigma}(L)} \setminus \widehat{M}_{n\Sigma}(L) \right) =0\,.
\ee
The continuity property of measures $\mu$ says that, for every decreasing 
sequence $A_n$ of sets with $\bigcap_n A_n =: A_\infty$, $\mu(A_n) \to \mu(A_\infty)$ as $n\to\infty$. For every $\Psi_\Sigma \in \Hilbert_\Sigma$, $\mu(\cdot) := \langle \Psi_\Sigma| P_\Sigma(\cdot) |\Psi_\Sigma \rangle$ is a measure. We know  from Lemma~\ref{lem:setestimate} that $\widehat{M}_{n\Sigma}(L) \subseteq \widecheck{M}_{n\Sigma}(L)$.

We show that for every $L\in\{0,1\}^r$, the sequence $A_n := \widecheck{M}_{n\Sigma}(L) \setminus \widehat{M}_{n\Sigma}(L)$ is decreasing: It suffices to show that $\widecheck{M}_{n\Sigma}(L)$ is decreasing and $ \widehat{M}_{n\Sigma}(L)$ is increasing. We know from Lemma~\ref{lem:closure} 
that $\widecheck{C}_{n\ell}$ is decreasing and $\widehat{C}_{n\ell}$ is increasing, so by \eqref{inclusionexists}, both $\exists\widecheck{C}_{n\ell}$ and $\emptyset\widehat{C}_{n\ell}$ are decreasing, so $\widecheck{M}_{n\ell}(L_\ell)$ (which is either $\exists\widecheck{C}_{n\ell}$ or $\emptyset\widehat{C}_{n\ell}$, depending on $L_\ell$) is decreasing, and so is
\be
\widecheck{M}_{n\Sigma}(L) = \bigcap_{\ell=1}^r \widecheck{M}_{n\ell}(L_\ell)\,.
\ee
Likewise, $\widehat{M}_{n\ell}(L_\ell)$ (which is either $\exists\widehat{C}_{n\ell}$ or $\emptyset\widecheck{C}_{n\ell}$, depending on $L_\ell$) is increasing, and so is $\widehat{M}_{n\Sigma}(L)$. Therefore, $A_n$ is decreasing, as claimed.

We can conclude that
\be
\widecheck{\Prob}_n(L) - \widehat{\Prob}_n(L) =
\langle \Psi_\Sigma |P_\Sigma \bigl(\widecheck{M}_{n\Sigma}(L) \setminus \widehat{M}_{n\Sigma}(L) \bigr)|\Psi_\Sigma \rangle \to 0 \text{ as }n\to\infty.
\ee
This establishes the desired squeeze theorem argument and finishes the proof of Theorem~\ref{thm:1}.
\end{proof}

\begin{proof}[Proof of Corollary~\ref{cor:strong}.]
It is well known that for a sequence $P_n$ of projections, weak convergence to the projection $P$ (i.e., $\langle \Psi|P_n|\Psi\rangle \to\langle \Psi|P|\Psi\rangle$ for every $\Psi$) implies strong convergence (i.e., $P_n\Psi\to P\Psi$ for every $\Psi$).\footnote{For the sake of completeness, here is a proof: First, $P_n^2=P_n$ and $P^2=P$ imply that $\|P_n\Psi\|^2=\langle \Psi|P_n^2|\Psi\rangle = \langle \Psi|P_n|\Psi\rangle 
\to \langle\Psi|P|\Psi\rangle= \|P\Psi\|^2$. Second, since $\langle \Psi |S|\Phi\rangle$ can be expressed through $\langle \Psi\pm \Phi|S|\Psi\pm \Phi\rangle$ and $\langle \Psi \pm i\Phi|S|\Psi \pm i\Phi\rangle$ (polarization identity \cite[p.~63]{RS1}), weak convergence implies $\langle \Psi|P_n|\Phi\rangle \to \langle \Psi|P|\Phi\rangle$ for every $\Psi$ and $\Phi$. Now $\|P_n\Psi - P\Psi\|^2=\langle\Psi|(P_n-P)^2| \Psi\rangle=\langle\Psi|P_n^2-P_nP-PP_n+P^2 |\Psi\rangle=\|P_n\Psi\|^2-\langle \Psi|P_n| P\Psi\rangle - \langle P\Psi|P_n|\Psi\rangle +\|P\Psi\|^2 \to \|P\Psi\|^2-\langle \Psi| P|P\Psi\rangle - \langle P\Psi|P|\Psi\rangle +\|P\Psi\|^2=0$.
} Set $P_n=U_{\Upsilon_n}^\Sigma P_{\Upsilon_n}(M_{\sB_n}(L)) U_\Sigma^{\Upsilon_n}$ and $P=P_{\Sigma}(M_{\sP}(L))$. Then Theorem~\ref{thm:1} provides the weak convergence, and the strong convergence was what we claimed.
\end{proof}

\paragraph{Remarks.}

\begin{enumerate}
\setcounter{enumi}{\theremarks}
\item {\it Type of convergence of $(\Upsilon_n)_{n\in\NNN}$.} The proof of Theorem~\ref{thm:1} still goes through unchanged if the convergence of the sequence $(\Upsilon_n)_{n\in\NNN}$ is not uniform but uniform on every bounded set.
\item {\it Alternative definition of $B_{n\ell}$.}\label{rem:Bnell} In order to avoid the choice of a particular Lorentz frame in the definition of $B_{n\ell}$ and thus of the detection probabilites, we could replace $B_{n\ell}$ by
\be
\widecheck{B}_{n\ell} := \Sr(P_\ell,\Upsilon_n) \,.
\ee
(The use of $\Gr$ instead of $\Sr$ would lead to overlap among the $B_{n\ell}$, so they would no longer form a partition.) With this change, Theorem~\ref{thm:1} remains valid. In the proof, we then need to modify the definition of $\widehat{C}_{n\ell}$ to
\be
\widehat{C}_{n\ell} := \Sr(\widecheck{B}_{n\ell},\Sigma)\, ,
\ee
while the definition of $\widecheck{C}_{n\ell}$ is kept as it is. We would still use a preferred Lorentz frame for the definition of $\widecheck{C}_{n\ell}$, but that is a matter of the method of proof, not of the statement of the theorem. The proof goes through as before, except that 
\eqref{Pinterior} needs to be checked anew: it is still true because for every $x$ in the 3-interior of $P_\ell$, $\Gr(\Gr(x,\Upsilon_n),\Sigma)\subset P_\ell$ for sufficiently large $n$.
\end{enumerate}
\setcounter{remarks}{\theenumi}

\bigskip

\noindent\textit{Acknowledgments.}
We thank Matthias Lienert and Stefan Teufel for helpful discussions.
This work was financially supported by the Wilhelm Schuler-Stiftung T\"ubingen, by DAAD (Deutscher Akademischer Austauschdienst) and also by the Basque Government through the BERC 2018-2021 program and by the Ministry of Science, Innovation and Universities: BCAM Severo Ochoa accreditation SEV-2017-0718.

\end{document}